\documentclass[a4paper,UKenglish]{lipics-v2018}

\usepackage{microtype}
\usepackage{xcolor}

\title{Semi-Online Bipartite Matching}

\titlerunning{Semi-Online Bipartite Matching}

\author{Ravi Kumar}{Google, Mountain View, CA, USA}{ravi.k53@gmail.com}{}{}
\author{Manish Purohit}{Google, Mountain View, CA, USA}{mpurohit@google.com}{}{}
\author{Aaron Schild}{University of California, Berkeley, CA, USA}{aschild@berkeley.edu}{}{}
\author{Zoya Svitkina}{Google, Mountain View, CA, USA}{zoya@cs.cornell.edu}{}{}
\author{Erik Vee}{Google, Mountain View, CA, USA}{erikvee@google.com}{}{}

\authorrunning{R. Kumar, M. Purohit, A. Schild, Z. Svitkina, and E. Vee}

\Copyright{Ravi Kumar, Manish Purohit, Aaron Schild, Zoya Svitkina, and Erik Vee}

\subjclass{Theory of Computation $\rightarrow$ Design and Analysis of Algorithms $\rightarrow$ Online Algorithms}


\category{}

\relatedversion{}

\supplement{}

\funding{}

\acknowledgements{We thank Michael Kapralov for introducing us to the bipartite matching skeleton decomposition in~\cite{goel2012communication}.
}

\EventEditors{}
\EventNoEds{0}
\EventLongTitle{}
\EventShortTitle{}
\EventAcronym{}
\EventYear{}
\EventDate{}
\EventLocation{}
\EventLogo{}
\SeriesVolume{}
\ArticleNo{} 
\nolinenumbers 
\hideLIPIcs  

\usepackage{algorithm2e}

\newcommand{\size}[1]{\ensuremath{|#1|}}
\newcommand\calS{\mathcal{S}}
\newcommand\calD{\mathcal{D}}
\newcommand{\setpair}[1]{#1}
\newcommand{\ignore}[1]{}

\begin{document}

\maketitle

\begin{abstract}
In this paper we introduce the \emph{semi-online} model that generalizes the classical online computational model. The semi-online model postulates that the unknown future has a predictable part and an adversarial part; these parts can be arbitrarily interleaved.  An algorithm in this model operates as in the standard online model, i.e., makes an irrevocable decision at each step.

We consider bipartite matching in the semi-online model, for both integral and fractional cases.  Our main contributions are competitive algorithms for this problem that are close to or match a hardness bound.  The competitive ratio of the algorithms nicely interpolates between the truly offline setting (no adversarial part) and the truly online setting (no predictable part).
\end{abstract}

\section{Introduction}
\label{sec:intro}

Modeling future uncertainty in data while ensuring that the model remains both realistic and tractable has been a formidable challenge facing the algorithms research community.  One of the more popular, and reasonably realistic, such models is the online computational model.  In its classical formulation, data arrives one at a time and upon each arrival, the algorithm has to make an irrevocable decision agnostic of future arrivals.  Online algorithms boast a rich literature and problems such as caching, scheduling, and matching---each of which abstracts common practical scenarios---have been extensively investigated~\cite{borodin2005online, mehta2013online}.   Competitive analysis, which measures how well an online algorithm performs compared to the best offline algorithm that knows the future, has been a linchpin in the study of online algorithms. 

While online algorithms capture some aspect of the future uncertainty in the data, the notion of competitive ratio is inherently worst-case and hence the quantitative guarantees it offers are often needlessly pessimistic.  A natural question that then arises is: how can we avoid modeling the worst-case scenario in online algorithms?  Is there a principled way to incorporate some knowledge we have about the future?  There have been a few efforts trying to address this point from different angles.  One line of attack has been to consider oracles that offer some advice on the future; such oracles, for instance, could be based on machine-learning methods.  This model has been recently used to improve the performance of online algorithms for reserve price optimization, caching, ski rental, and scheduling~\cite{sergei1, sergei2, KPS18}.  Another line of attack posits a distribution on the data~\cite{BS12,MNS12,MGZ12} or the arrival model; for instance, random arrival models have been popular in online bipartite matching and are known to beat the worst-case bounds~\cite{karande2011online,mahdian2011online}.  A different approach is to assume a distribution on future inputs; the field of stochastic online optimization focuses on this setting~\cite{RussellPascal}.  The advice complexity model, where the partial information about the future is quantified as advice bits to an online algorithm, has been studied as well in complexity theory~\cite{advice}.

In this work we take a different route.  At a very high level, the idea is to tease the future data apart into a predictable subset and the remaining adversarial subset.  As the names connote, the algorithm can be assumed to know everything about the former but nothing about the latter.  Furthermore, the predictable and adversarial subsets can arrive arbitrarily interleaved yet the algorithm still has to operate as in the classical online model, i.e., to make irrevocable decisions upon each arrival.   Our model thus offers a natural interpolation between the traditional offline and online models;  we call ours the \emph{semi-online} model.  Our goal is to study algorithms in the semi-online model and to analyze their competitive ratios; unsurprisingly, the bounds depend on the size of the adversarial subset. Ideally, the competitive ratio should approach the offline optimum bounds if the adversarial fraction is vanishing and should approach the online optimum bounds if the predictable fraction is vanishing.  

\subparagraph{Bipartite Matching.}

As a concrete problem in the semi-online setting, we focus on bipartite matching, in both its integral and fractional versions. In the well-known online version of the problem, which is motivated by online advertising, there is a bipartite graph with an offline side that is known in advance and an online side that is revealed one node at a time together with its incident edges.  In the semi-online model, the nodes in the online side are partitioned into a predicted set of size $n-d$ and an adversarial set of size $d$.  The algorithm knows the incident edges of all the nodes in the former but nothing about the nodes in the latter.  We can thus also interpret the setting as online matching with partial information and predictable uncertainty (pardon the oxymoron).
In online advertising applications, there are many predictably unpredictable events.  For example, during the soccer world cup games, we know the nature of web traffic will be unpredictable but nothing more, since the actual characteristics will depend on how the game progresses and which team wins. 

We also consider a variant of semi-online matching in which the algorithm does not know which nodes are predictable and which are adversarial.  In other words, the algorithm receives a prediction for all online nodes, but the predictions are  correct only for some $n-d$ of them.  We call this the \emph{agnostic} case.

\subparagraph{Main Results.}

We initially assume that the optimum solution on the bipartite graph, formed by the offline nodes on one side and by the predicted and adversarial nodes on the other, is a perfect matching (this is later extended to the general case).
We present two algorithms and a hardness result for the integral semi-online bipartite matching problem.  Let $\delta = d/n$ be the fraction of adversarial nodes.  The Iterative Sampling algorithm, described in Section~\ref{sec:alg1}, obtains a competitive ratio of $(1 - \delta + \frac{\delta^2}{2}(1-1/e))$. (Observe that an algorithm that ignores all the adversarial nodes and outputs a maximum matching in the predicted graph achieves a competitive ratio of only $1-\delta$.)  Our algorithm ``reserves'' a set of offline nodes to be matched to the adversarial nodes by repeatedly selecting a random offline node that is unnecessary for matching the predictable nodes.  It is easy to see that algorithms that deterministically reserve a set of offline nodes can easily be thwarted by the adversary.

The second algorithm, described in Section \ref{sec:structured-sampling}, achieves an improved competitive ratio of $(1-\delta+\delta^2(1-1/e))$. This algorithm samples a maximum matching in the predicted graph by first finding a matching skeleton~\cite{goel2012communication,LeeSingla17} and then sampling a matching from each component in the skeleton using the dependent rounding scheme of~\cite{gandhi2006dependent}. 
This allows us to sample a set of offline nodes that, in expectation, has a large overlap with the set matched to adversarial nodes in the optimal solution.
Surprisingly, in Section \ref{sec:sets-alg}, we show that it is possible to sample from arbitrary set systems so that the same ``large overlap'' property is maintained.
We prove the existence of such distributions using LP duality and believe that this result may be of independent interest.  

We next consider the fractional version of the semi-online bipartite matching problem, for which we obtain a competitive ratio of $(1 - \delta e^{-\delta}) \approx (1 - \delta + \delta^2 - \delta^3/2 + \cdots)$ using the matching skeleton decomposition combined with a primal-dual analysis of the water-filling algorithm. Note that this expression coincides with the best offline bound (i.e, $1$) and the best online bound (i.e., $1-1/e$) at the extremes of $\delta=0$ and $\delta=1$, respectively.

For both the integral and fractional cases, a hardness result that appears in~\cite[Theorem 3.23]{esfandiari2018}  for a model with stochastic predictions applies to our semi-online model as well. It shows a hardness of $1 - \delta e^{-\delta}$, implying that our algorithm for the fractional case is tight, and the ones for the integral case are nearly tight. We conjecture $1 - \delta e^{-\delta}$ to be the optimal bound for the integral problem as well. 


\subparagraph{Extensions.}

In Section~\ref{sec:extensions}, we explore variants of the semi-online matching model, including the agnostic version and fractional matchings, and present upper and lower bounds in those settings.  To illustrate the generality of our semi-online model, we consider a semi-online version of the classical ski rental problem.  In this version, the skier knows whether or not she'll ski on certain days while other days are uncertain. Interestingly, there is an algorithm with a competitive ratio of the same form as our hardness result for matchings, namely $1 - (1-x)e^{-(1-x)}$, where $(1-x)$ is a parameter analogous to $\delta$ in the matching problem.  We wonder if this form plays a role in semi-online algorithms similar to what $(1-1/e)$ has in many online algorithms~\cite{karlin2003stories}.

\subparagraph{Other Related Work.}

The use of (machine learned) predictions for revenue optimization problems was first proposed in~\cite{sergei1}.  The concepts were formalized further and applied to online caching in~\cite{sergei2} and ski rental and non-clairvoyant scheduling in~\cite{KPS18}. Online matching with forecasts was first studied in~\cite{vee2010matching}; however, that paper is on forecasting the demands rather than the structure of the graph as in our case.  The problem of online matching when edges arrive in batches was considered in~\cite{LeeSingla17} where a $(1/2 + 2^{-O(s)})$-competitive ratio is shown, with $s$ the number of stages.
However, the batch framework differs from ours in that in our case, the nodes arrive one at a time and are arbitrarily interleaved.  Online allocation problem in a setting that interpolates between stochastic and adversarial models has also been studied~\cite{esfandiari2018}.
 
There has been a lot of work on online bipartite matching and its variants. The RANKING algorithm~\cite{karp1990optimal} selects a random permutation of the offline nodes and then matches each online node to its unmatched neighbor that appears earliest in the permutation. It is well-known to obtain a competitive ratio of $(1 - 1/e)$, which is best possible. For a history of the problem and significant advances, see the monograph~\cite{mehta2013online}.  The ski rental problem has also been extensively studied; the optimal randomized algorithm has ratio $e/(e-1)$~\cite{karlin1994competitive}.  The term ``semi-online'' has been used in scheduling when an online scheduler knows the sum of the jobs’ processing times (e.g., see~\cite{semionline1}) and in online bin-packing when a lookahead of the next few elements is available (e.g., see~\cite{semionline2}); our use of the term is more quantitative in nature.

\section{Model}
\label{sec:model}

We now formally define the \emph{semi-online bipartite matching} problem.
We have a bipartite graph $G = (\setpair{U, V}, E_G)$ where $U$ is the set of nodes available offline and nodes in $V$ arrive online. Further, the online set $V$ is partitioned into the \emph{predicted} nodes $V_P$ and the \emph{adversarial} nodes $V_A$. The \emph{predicted graph} $H = (\setpair{U, V_P}, E_H)$ is the subgraph of $G$ induced on the nodes in $U$ and $V_P$. Initially, an algorithm only knows $H$ and is unaware of edges between $U$ and $V_A$.
 The algorithm is allowed to preprocess $H$ before any online node actually arrives. In the online phase, at each step, one node of $V$ is revealed with its incident edges, and has to be either irrevocably matched to some node in $U$ or abandoned. Nodes of $V$ are revealed in an arbitrary order\footnote{The arrival order can be adversarial, including interleaving the nodes in $V_P$ and $V_A$.} and the process continues until all of $V$ has been revealed. 

We note that when a node $v\in V$ is revealed, the algorithm can ``recognize'' it by its edges, i.e., if there is some node $v' \in V_P$ that has the same set of neighbors as $v$ and has not been seen before, then $v$ can be assumed to be $v'$. There could be multiple identical nodes in $V_P$, but it is not important to distinguish between them. If an online node comes that is not in $V_P$, then the algorithm can safely assume that it is from $V_A$. (In Section \ref{sec:extensions}, we consider a model where the predicted graph can have errors and hence this assumption is invalid.)

We introduce a quantity $\delta$ to measure the level of knowledge that the algorithm has about the input graph $G$. Competitive ratios that we obtain are functions of $\delta$. For any graph $I$, let $\nu(I)$ denote the size of the maximum matching in $I$. Then we define $\delta = \delta(G) = 1 - \frac{\nu(H)}{\nu(G)}$. Intuitively, the closer $\delta$ is to $0$, the more information the predicted graph $H$ contains and the closer the instance is to an offline problem. Conversely, $\delta$ close to $1$ indicates an instance close to the pure online setting. Note that the algorithm does not necessarily know $\delta$, but we use it in the analysis to bound the competitive ratio. 
For convenience, in this paper we assume that the input graph $G$ contains a perfect matching. Let $n = |U| = |V|$ be the number of nodes on each side and $d = |V_A|$ be the number of adversarial online nodes.
In this case, $\delta$ simplifies to be the fraction of online nodes that are adversarial, i.e., $\delta = \frac{\size{V_A}}{\size{V}} = \frac{d}{n}$. In Appendix \ref{sec:app_imperfect}, we extend our techniques to handle the general version of the problem when $G$ may not contain a perfect matching.

\section{Iterative Sampling Algorithm}
\label{sec:alg1}

In this section we give a simple polynomial time algorithm for bipartite matching in the semi-online model. We describe the algorithm in two phases: a \emph{preprocessing phase} that finds a maximum matching $M$ in the predicted graph $H$ and an \emph{online phase} that processes each node upon its arrival to find a matching in $G$ that extends $M$.

\paragraph*{Preprocessing Phase} The goal of the preprocessing phase is to find a maximum matching in the predicted graph $H$. However, if we deterministically choose a matching, the adversary can set up the neighborhoods of $V_A$ so that all the neighbors of $V_A$ are used in the chosen matching, and hence the algorithm is unable to match any node from $V_A$.
Algorithm \ref{alg:preprocessing} describes our algorithm to sample a (non-uniform) random maximum matching from $H$.

\begin{algorithm}
\SetAlgoVlined
\newcommand{\forcond}{$i=0$ \KwTo $n$}
\SetKwFunction{FRecurs}{Preprocess}%
\SetKwProg{Fn}{Function:}{\string:}{}
\Fn(){\FRecurs{H}}{
\KwData{Predicted graph $H$}
\KwResult{Maximum matching in $H$, a sequence of nodes from $U$}
\BlankLine
Let $H_0 \leftarrow H, U_0 \leftarrow U$ \\
\For{$i = 1, 2, \ldots, d$}{
$U_i \leftarrow \{u \in U_{i-1} ~\mid~ \nu(H_{i-1} \setminus \{u\}) = n - d\}$ \tcc*[f]{set of nodes whose removal does not change the size of the maximum matching} \\
Let $u_i$ be a uniformly random node in $U_i$ \\
$H_i \gets H_{i-1}\setminus \{u_i\}$  
} 
$M \leftarrow $ Arbitrary maximum matching in $H_{d}$ \\
$R \leftarrow $  Uniformly random permutation of $\{u_1, \ldots, u_d\}$ \\
\KwRet{M, R}
}
\caption{Iterative Sampling: Preprocessing Phase.}
\label{alg:preprocessing}
\end{algorithm}

\paragraph*{Online Phase} In the online phase nodes from $V$ arrive one at a time and we are required to maintain a matching such that the online nodes are either matched irrevocably or dropped. In this phase, we match the nodes in $V_P$ as per the matching $M$ obtained from the preprocessing phase, i.e., we match $v \in V_P$ to node $M(v) \in U$ where $M(v)$ denotes the node matched to $v$ by matching $M$. The adversarial nodes in $V_A$ are matched to nodes in $R$ that are not used by $M$ using the RANKING algorithm~\cite{karp1990optimal}. Algorithm~\ref{alg:online} describes the complete online phase of our algorithm.

\begin{algorithm}
\SetAlgoVlined
\newcommand{\forcond}{$i=0$ \KwTo $n$}
\SetKwFunction{FRecurs}{Preprocess}%
\SetKwProg{Fn}{Function:}{\string:}{}

\BlankLine
$M, R \leftarrow \FRecurs{H}$ \\
\For{$v \in V$ arriving online}
{
\eIf(\tcc*[f]{predicted node}) {$v \in V_P$} 
{Match $v$ to $M(v)$} 
(\tcc*[f]{adversarial node})
{Match $v$ to the first unmatched neighbor in $R$, if one exists \tcc*[f]{RANKING}}
}
\caption{Online Phase.}
\label{alg:online}
\end{algorithm}

\paragraph*{Analysis} For the sake of analysis, we construct a sequence of matchings $\{M_i^*\}_{i=0}^d$ as follows.
Let $M_0^*$ be an arbitrary perfect matching in $G$. 
For $i \geq 1$, by definition of $U_i$, there exists a matching $M'_i$ in $H_i$ of size $n-d$ that does not match node $u_i$.
Hence, $M'_i \cup M_{i-1}^*$ is a union of disjoint paths and cycles such that $u_i$ is an endpoint of a  path $P_i$. Let $M^*_i = M^*_{i-1} \oplus P_i$, i.e., obtain $M^*_{i}$ from $M^*_{i-1}$ by adding and removing alternate edges from $P_i$.
It's easy to verify that $M^*_i$ is indeed a matching and $|M^*_i| \geq |M^*_{i-1}| - 1$. Since $|M_0^*| = n$, this yields $|M^*_{i}| \geq n - i, \ \forall\ 0 \leq i \leq d$. Further, by construction, $M^*_i$ does not match any nodes in $\{u_1, \ldots, u_i\}$.

\begin{lemma}
\label{lem:matched-adversary}
For all $0 \leq i \leq d$, all nodes $v \in V_P$ are matched by $M^*_i$. Further, $|M^*_i(V_A)| \geq d - i$, i.e., at least $d - i$ adversarial nodes are matched by $M^*_i$.
\end{lemma}

\begin{proof}
We prove the claim by induction. Since $M^*_0$ is a perfect matching, the base case is trivially true. By the induction hypothesis, we assume that $M^*_{i-1}$ matches all of $V_P$. 
Recall that $M'_i$ also matches all of $V_P$ and $M^*_i = M^*_{i-1} \oplus P_i$ where $P_i$ is a maximal path in $M'_i \cup M^*_{i-1}$. Since each node $v \in V_P$ has degree two in $M'_i \cup M^*_{i-1}$, $v$ cannot be an end point of $P_i$. Hence, all nodes $v \in V_P$ remain matched in $M^*_i$. Further, we have $|M^*_i(V_A)| = |M^*_i| - |M^*_i(V_P)| \geq (n-i) - (n-d) = d - i$ as desired.
\end{proof}

Equipped with the sequence of matchings $M^*_i$, we are now ready to prove that, in expectation, a large matching exists between the set $R$ of nodes left unmatched by the preprocessing phase and the set $V_A$ of adversarial nodes.

\begin{lemma}
\label{lem:large-reserved}
$\mathbb{E}[\nu(G[R\cup V_A])] \ge d^2/(2n)$ where $G[R\cup V_A]$ is the graph induced by the reserved vertices $R$ and the adversarial vertices $V_A$.
\end{lemma}

\begin{proof}
We construct a sequence of sets of edges $\{N_i\}_{i=0}^d$ as follows. Let $N_0 = \emptyset$. If $M_{i-1}^*(u_i)\in V_A$, let $e_i = \{u_i, M_{i-1}^*(u_i)\}$ be the edge of $M_{i-1}^*$ incident with $u_i$ and let $N_i = N_{i-1}\cup \{e_i\}$. Otherwise, let $N_i = N_{i-1}$. In other words, if the node $u_i$ chosen during the $i$th step is matched to an adversarial node by the matching $M^*_{i-1}$, add the matched edge to set $N_i$.

We show by induction that $N_i$ is a matching for all $i\ge 0$. $N_0$ is clearly a matching. When $i > 0$, either $N_i = N_{i-1}$ (in which case we are done by the inductive hypothesis), or $N_i = N_{i-1}\cup \{e_i\}$. Let $e_i = (u_i, v_i)$ and consider any other edge $e_j = (u_j, v_j) \in N_{i-1}$. Since $u_j \notin H_{i-1}$, we have $u_j \neq u_i$. By definition, node $v_i$ is matched in $M^*_{i-1}$. By construction, this implies that $v_i$ must be matched in all previous matchings in this sequence, in particular, $v_i$ must be matched in $M^*_j$ (since a node $v \in V_A$ that is unmatched in $M^*_{k-1}$ can never be matched by $M^*_k$). However, since $v_j = M^*_{j-1}(u_j)$, the matching $M^*_j = M^*_{j-1} \setminus \{e_j\}$ and hence $v_j$ is not matched in $M^*_j$. Hence $v_i \neq v_j$. Thus we have shown that $e_i$ does not share an endpoint with any $e_j \in N_{i-1}$ and hence $N_i$ is a matching.

By linearity of expectation we have the following.
\begin{align*}
    \mathbb{E}[|N_i|] \enspace &= \enspace \mathbb{E}[|N_{i-1}|] + \Pr_{u_i}[ M^*_{i-1}(u_i) \in V_A].
\intertext{However, by Lemma \ref{lem:matched-adversary}, since $M^*_{i-1}$ matches all of $V_P$, we must have $M^*_{i-1}(V_A) \subseteq U_i$. Hence,}
    \mathbb{E}[|N_i|] \enspace &\geq \enspace \mathbb{E}[|N_{i-1}|] + \frac{|M^*_{i-1}(V_A)|}{|U_i|} 
    \enspace \geq \enspace \mathbb{E}[|N_{i-1}|] + \frac{d-(i-1)}{n}.
    \intertext{Solving the recurrence with $|N_0| = 0$ gives}
    \mathbb{E}[|N_d|] \enspace &\geq \enspace \sum_{i=1}^{d} \frac{i}{n} \enspace \geq \enspace \frac{d(d+1)}{2n}.
\end{align*}
The lemma follows since $N_d$ is a matching in $G[R \cup V_A]$.
\end{proof}

\begin{theorem}\label{thm:algo-non-agnostic-weaker}
There is a randomized algorithm for the semi-online bipartite matching problem with a competitive ratio of at least $(1 - \delta + (\delta^2/2)(1-1/e))$ in expectation when the input graph $G$ has a perfect matching.
\end{theorem}

\begin{proof}
Algorithm \ref{alg:preprocessing} guarantees that the matching $M$ found in the preprocessing phase matches all predicted nodes and has size $n-d = n(1-\delta)$. Further, in the online phase, we use the RANKING~\cite{karp1990optimal} algorithm on the graph $G[R \cup V_A]$. Since RANKING is $(1-1/e)$-competitive, the expected number of adversarial nodes matched is at least $(1 - 1/e) \nu(G[R \cup V_A])$. By Lemma \ref{lem:large-reserved}, this is at least $(1-1/e)(\frac{d^2}{2n}) = (\delta^2 n / 2)(1-1/e)$.

Therefore, the total matching has expected size $n(1 - \delta + (\delta^2 / 2)(1-1/e))$ as desired.
\end{proof}

Using a more sophisticated analysis, we can show that the iterative sampling algorithm yields a tighter bound of $(1 - \delta + \delta^2/2 - \delta^3/2)$. However we omit the proof because the next section presents an algorithm with an even better guarantee.

\section{Structured Sampling}
\label{sec:structured-sampling}

In this section we give a polynomial time algorithm for the semi-online bipartite matching that yields an improved competitive ratio of $(1 - \delta + \delta^2(1 - 1/e))$. We first discuss the main ideas in Section \ref{sec:intuition} and then describe the algorithm and its analysis in Section \ref{sec:sampling-proof}.

\begin{theorem}
\label{thm:main-alg}
There is a randomized algorithm for the semi-online bipartite matching problem with a competitive ratio of at least $(1-\delta+\delta^2(1-1/e))$ in expectation when the input graph $G$ has a perfect matching.
\end{theorem}

\subsection{Main Ideas and Intuition}
\label{sec:intuition}

As with the iterative sampling algorithm, we randomly choose a matching of size $n-d$ (according to some distribution), and define the \emph{reserved} set $R$ to be the set of offline nodes that are not matched. As online nodes arrive, we follow the matching for the predicted nodes; for adversarial nodes, we run the RANKING algorithm on the reserved set $R$. 

Let $M^*$ be a perfect matching in $G$. For a set  of nodes $S$, let $M^*(S)$ denote the set of nodes matched to them by $M^*$.
Call a node $u \in U$ \emph{marked} if it is in $M^*(V_A)$, i.e., it is matched to an adversarial node by the optimal solution.
We argue that the number of marked nodes in the set $R$ chosen by our algorithm is at least $d^2/n$ in expectation. Since RANKING finds a matching of at least a factor $(1-1/e)$ of optimum in expectation, this means that we find a matching of size at least $d^2 / n \cdot (1-1/e)$ on the reserved nodes in expectation. Combining this with the matching of size $n-d$ on the predicted nodes, this gives a total  of  $n-d + d^2 / n \cdot (1-1/e) = n (1 - \delta + \delta^2(1-1/e))$.

The crux of the proof lies in showing that $R$ contains many marked nodes. 
Ideally, we would like to choose a random matching of size $n-d$ in such a way that  each node of $U$ has probability $d/n$ of being in $R$. Since there are $d$ marked nodes total, $R$ would contain $d^2/n$ of them in expectation. However, such a distribution over matchings does not always exist.

Instead, we use a graph decomposition to guide the sampling process. The marginal probabilities for nodes of $U$ to be in $R$ may differ, but nevertheless $R$ gets the correct total number of marked nodes in expectation.  
$H$ is decomposed into bipartite pairs $(S_i, T_i)$, with $|S_i| \leq |T_i|$, so that the sets $S_i$ partition $V_P$ and the sets $T_i$ partition $U$. This decomposition allows one to choose a random matching between $S_i$ and $T_i$ of size $|S_i|$ so that each node in $T_i$ is reserved with the same probability. Letting $n_i = |T_i|$ and $d_i = |T_i| - |S_i|$, this probability is precisely $d_i/n_i$.
Finally, we argue that the adversary can do no better than to mark $d_i$ nodes in $T_i$, for each $i$. Hence, the expected number of nodes in $R$ that are marked is at least $\sum_i (d_i^2/n_i)$, which we  lower bound by $d^2/n$.

\subsection{Proof of Theorem \ref{thm:main-alg}}
\label{sec:sampling-proof}
We decompose the graph $H$ into more structured pieces using a construction from~\cite{goel2012communication} and utilize the key observation that the decomposition implies a \emph{fractional matching}. Recall that a {fractional matching} is a function $f$ that assigns a value in $[0,1]$ to each edge in a graph, with the property that $\sum_{e \ni v} f(e) \leq 1$ for all nodes $v$. The quantity $\sum_{e \ni v} f(e)$ is referred to as the \emph{fractional degree} of $v$. We use $\Gamma(S)$ to denote the set of neighbors of nodes in $S$.

\begin{lemma}\label{lem:skeleton} [Restatement of Lemma 2 from \cite{LeeSingla17}]
Given a bipartite graph $H = (\setpair{U, V_P}, E_H)$ with $|U| \geq |V_P|$ and a maximum matching of size $|V_P|$, there exists a partition of $V_P$ into sets $S_0$, \ldots, $S_m$ and a partition of $U$ into sets $T_0$, \ldots, $T_m$ for some $m$
such that the following holds:
\begin{itemize}
    \item $\Gamma(\bigcup_{i < j} S_i) = \bigcup_{i<j} T_i$ for all $j$.
    \item For all $i < j$, $\frac{|S_i|}{|T_i|} > \frac{|S_j|}{|T_j|}$.
    \item There is a fractional matching in $H$ of size $|V_P|$, where for all $i$, the fractional degree of each node in $S_i$ is 1 and the fractional degree of each node in $T_i$ is $|S_i| / |T_i|$. In this matching, nodes in $S_i$ are only matched with nodes in $T_i$ and vice versa.
\end{itemize}
Further, the $(S_i, T_i)$ pairs can be found in polynomial time.
\end{lemma}

In \cite{goel2012communication} and \cite{LeeSingla17}, the sets in the decomposition with $|S_i| < |T_i|$ are indexed with positive integers $i>0$, the sets with $|S_0| = |T_0|$ get an index of $0$, and the ones with $|S_i| > |T_i|$ get  negative indices $i<0$.
Under our assumption that $H$ supports a matching that matches all nodes of $V_P$, the decomposition does not contain sets with $|S_i| > |T_i|$, as the first such set would have $|S_i|>|\Gamma(S_i)|$, violating Hall's theorem. So we start the indices from $0$.

Equipped with this decomposition, we choose a random matching between $S_i$ and $T_i$ such that each node in $T_i$ is reserved\footnote{Recall that we say a node $u$ is reserved by an algorithm if $u$ is \emph{not} matched in the predicted graph $H$.} with the same probability. Since each $(S_i, T_i)$ pair has a fractional matching, the dependent randomized rounding scheme of~\cite{gandhi2006dependent} allows us to do exactly that.

\begin{lemma}\label{lem:rounding}
Fix an index $i$ and let $S_i, T_i$ be defined as in Lemma~\ref{lem:skeleton}. Then there is a distribution over matchings with size $|S_i|$ between $S_i$ and $T_i$ such that for all $u \in T_i$, the probability that the matching contains $u$ is $|S_i| / |T_i|$.
\end{lemma}

\begin{proof}
Given any bipartite graph $G'$ and a fractional matching over $G'$, the dependent rounding scheme of~\cite{gandhi2006dependent} yields an integral matching such that the probability that any node $v \in G'$ is matched exactly equals its fractional degree. Since Lemma \ref{lem:skeleton} guarantees a fractional matching such that the fractional degree of each node in $S_i$ is 1 and the fractional degree of each node in $T_i$  is $|S_i|/|T_i|$, the proof follows.
\end{proof}

We are now ready to complete the description of our algorithm. Algorithm \ref{alg:structured_preprocessing} is the preprocessing phase, while the online phase remains the same as earlier (Algorithm \ref{alg:online}). In the preprocessing phase, we  find a decomposition of the predicted graph $H$, and  sample a matching using Lemma \ref{lem:rounding} for each component in the decomposition. In the online phase, we match all predicted online nodes using the sampled matching and use  RANKING  to match the adversarial online nodes.

\begin{algorithm}
\DontPrintSemicolon
\SetAlgoVlined
\newcommand{\forcond}{$i=0$ \KwTo $n$}
\SetKwFunction{FRecurs}{Preprocess}%
\SetKwProg{Fn}{Function:}{\string:}{}
\Fn(){\FRecurs{H}}{
\KwData{Predicted graph $H$}
\KwResult{Maximum matching in $H$, sequence of nodes from $U$}
\BlankLine
Decompose $H$ into $\{(S_i, T_i)\}_{i=0}^m$ pairs using Lemma~\ref{lem:skeleton}\;
$M_i \leftarrow $ Random matching between $S_i$ and $T_i$ using Lemma~\ref{lem:rounding}\;
$\displaystyle M \leftarrow \bigcup_i M_i$\;
Let $R_{\text{set}} \subseteq U$ be the set of nodes unmatched by $M$\;
$R \leftarrow$ Uniform random permutation of $R_{\text{set}}$\;
\KwRet{M, R}
}
\caption{Structured Sampling: Preprocessing Phase.}
\label{alg:structured_preprocessing}
\end{algorithm}

Let $n_i = |T_i|$, $d_i = |T_i| - |S_i|$, and let $R_i = R \cap T_i$ be the set of reserved nodes in $T_i$. Then Lemma~\ref{lem:rounding} says that each node in $T_i$ lands in $R_i$ with probability $d_i/n_i$ (although not independently). 
We now argue in Lemmas~\ref{lem:adv_matches} and~\ref{lem:sequence} that the adversary can do no better than to choose $d_i$ marked nodes in each $T_i$.

\begin{lemma}\label{lem:adv_matches}
Let $\ell_i = \size{M^*(V_A) \cap T_i}$, i.e., let $\ell_i$ be the number of marked nodes in $T_i$. 
Then for all $t \geq 0$,
$$
\sum_{0 \leq i\leq t} \ell_i \leq  \sum_{0 \leq i\leq t} d_i.
$$
\ignore{\sum_{i\leq t} \Big(\size{T_i}-\size{S_i}\Big) =}
\end{lemma}
\begin{proof}
Fix $t$. Note that since the $T_i$'s partition $U$, we have $\sum_{i\leq t} \ell_i = |M^*(V_A) \cap \bigcup_{i\leq t} T_i|$. Furthermore, since $M^*$ is a perfect matching, we have that 
$$\Big|M^*(V_A) \cap \bigcup_{i\leq t} T_i\Big| + \Big|M^*(V_P)\cap \bigcup_{i\leq t} T_i\Big| = \Big|\bigcup_{i\leq t} T_i\Big|.$$
So let us consider $M^*(V_P)\cap \bigcup_{i\leq t} T_i$. By Lemma~\ref{lem:skeleton}, $\Gamma(\bigcup_{i\leq t} S_i) = \bigcup_{i\leq t} T_i$. Since $M^*$ is a perfect matching, every node in $\bigcup_{i\leq t} S_i$ must be matched to a node in $\bigcup_{i\leq t} T_i$. There are $\sum_{i\leq t} (n_i - d_i)$ nodes in $\bigcup_{i\leq t} S_i \subseteq V_P$, hence
$|M^*(V_P)\cap \bigcup_{i\leq t} T_i| \geq \sum_{i\leq t} (n_i - d_i)$. Putting this together,
\begin{align*}
\sum_{i\leq t} \ell_i = 
\size{M^*(V_A) \cap T_i} &= \Big|\bigcup_{i\leq t} T_i\Big| - \Big|M^*(V_P)\cap \bigcup_{i\leq t} T_i\Big| \\
    &\leq \sum_{i\leq t} n_i - \sum_{i\leq t} (n_i - d_i)
    = \sum_{i \leq t} d_i,
\end{align*}
as we wanted.

\ignore{
Fix $t \leq m$, and let $U' = U - \bigcup_{i\leq t} T_i$. Since there is a perfect matching in the realized graph $G$, Hall's Theorem guarantees that there must be at least $|U'| - |\Gamma_H(U')|$ marked nodes in $U'$ where $\Gamma_H(U')$ denotes the set of neighbors of  $U'$ in the predicted graph $H$. That is,
\begin{align*}
\sum_{i > t} \ell_i &\geq |U'| - |\Gamma_H(U')|
\intertext{But Lemma~\ref{lem:skeleton} tells us that $\Gamma(\bigcup_{i\leq t} S_i) = \bigcup_{i\leq t} T_i$, hence there is no edge between $U'$ and $\bigcup_{i\leq t} S_i$. That is,
$\Gamma_H(U') \subseteq V_P - \bigcup_{i\leq t} S_i$. Hence,}
|\Gamma_H(U')| 
	&\leq |V_P| - \sum_{i\leq t} |S_i| 
    = n-d-\sum_{i\leq t} (n_i-d_i)
\intertext{Further, $|U'| = |U| - |\bigcup_{i\leq t} T_i| = n - \sum_{i\leq t} n_i$.
Putting this together,}
\sum_{i > t} \ell_i 
	&\geq |U'| - |\Gamma_H(U')|\\
    &\geq n - \sum_{i\leq t} n_i - \bigg(n-d-\sum_{i\leq t} (n_i-d_i)\bigg)\\
    &= d - \sum_{i\leq t} d_i
\end{align*}
Recalling that $\sum_{i\leq m} \ell_i = d$, we see that
$\sum_{i\leq t} \ell_i = d - \sum_{i > t} \ell_i \leq \sum_{i\leq t} d_i$, as desired.
}
\end{proof}

\begin{lemma}\label{lem:sequence}
Let $0 < a_0 \leq a_1 \leq \cdots \leq a_m$ be a non-decreasing sequence of positive numbers, and $\ell_0, \ldots, \ell_m$ and $k_0, \ldots, k_m$  be non-negative integers, such that
$\sum_{i=0}^m \ell_i = \sum_{i=0}^m k_i$ and for all $t \leq m$,
$\sum_{i\leq t} \ell_i \leq \sum_{i \leq t} k_i$.
Then 
$$
\sum_{i=0}^m \ell_i a_i \geq \sum_{i=0}^m k_i a_i.
$$
\end{lemma}
\begin{proof}
We claim that for any fixed sequence $k_0, \ldots, k_m$, the minimum of the left-hand side ($\sum_i \ell_i a_i$) is attained when $\ell_i = k_i$ for all $i$. Suppose for contradiction that $\{\ell_i\}$ is the lexicographically-largest minimum-attaining assignment that is not equal to $\{k_i\}$ and let $j$ be the smallest index with $\ell_j \neq k_j$.
It must be that $\ell_j < k_j$ to satisfy $\sum_{i\leq j} \ell_i \leq \sum_{i \leq j} k_i$. Also, $\sum_{i=0}^m \ell_i = \sum_{i=0}^m k_i$ implies that $j < m$ and that there must be an index $j'>j$ such that $\ell_{j'} > k_{j'}$. Let $j'$ be the lowest such index.

Let $\ell'_i = \ell_i$ for all $i \notin \{j, j'\}$. Set $\ell'_j = \ell_j + 1$ and $\ell'_{j'} = \ell_{j'}-1$.
Notice that we still have $\sum_{i\leq t} \ell'_i \leq \sum_{i \leq t} k_i$ for all $t$ and $\sum_{i=0}^m \ell'_i = \sum_{i=0}^m k_i$, and $\{\ell'_i\}$ is lexicographically larger than $\{\ell_i\}$. In addition,
$$
\sum_i \ell'_i a_i 
	= \sum_i \ell_i a_i + a_j - a_{j'}
    \leq \sum_i \ell_i a_i,
$$
which is a contradiction.
\end{proof}

We need one last technical observation before the proof of the main result.
\begin{lemma}\label{lem:cauchy}
Let $d_i, n_i$ be positive numbers with $\sum_i d_i = d$ and $\sum_i n_i = n$. Then
$$
\sum_i \frac{d_i^2}{n_i} \geq \frac{d^2}{n}.
$$
\end{lemma}
\begin{proof}
We invoke Cauchy--Schwarz, with vectors $u$ and $v$ defined by $u_i = \frac{d_i}{\sqrt{n_i}}$ and $v_i = \sqrt{n_i}$. Since $||u||^2 \geq | u \cdot v | ^2 / ||v||^2$, 
the result follows.
\end{proof}
\begin{theorem}\label{thm:marked}
Choose reserved set $R$ according to  Algorithm \ref{alg:online}. Then the expected number of marked nodes in $R$ is at least $\delta^2 n$. That is, $|R \cap M^*(V_A)| \geq \delta^2 n$ in expectation.
\end{theorem}
\begin{proof}
As in Lemma~\ref{lem:adv_matches}, let $\ell_i = \size{M^*(V_A) \cap T_i}$. Again, we have $\sum_{i\leq t} \ell_i \leq \sum_{i\leq t} d_i$ for all $t$ and $\sum_{i\leq m} \ell_i = d = \sum_{i\leq m} d_i$.
For each $i$, the node $u\in T_i$ is chosen to be in $R$ with probability $d_i/n_i$, with the $d_i/n_i$ forming an increasing sequence. So the expected size of $|R \cap M^*(V_A)|$ is given by
\begin{align*}
\sum_i \frac{d_i}{n_i}\ell_i
	&\geq \sum_i \frac{d_i}{n_i} d_i &\mbox{\ by Lemma~\ref{lem:sequence}}\\
    &\geq \frac{d^2}{n} &\mbox{\ by Lemma~\ref{lem:cauchy}.}
\end{align*}
Since $\delta = d/n$, the proof follows.
\ignore{\sum_i \frac{|S_i|-|T_i|}{|S_i|}\ell_i = }
\end{proof}

\begin{proof}[Proof of Theorem \ref{thm:main-alg}]
The size of the matching, restricted to non-adversarial nodes, is $\sum_i (n_i -d_i) = n-d = n-\delta n$. Further, by Theorem \ref{thm:marked}, we have reserved at least $\delta^2 n$ nodes that can be matched to the adversarial nodes. RANKING will match at least a $(1-1/e)$ fraction of these in expectation. So in expectation, the total matching has size at least
$
n-\delta n + \delta^2 n(1-1/e) = n(1-\delta +\delta^2 (1-1/e))
$ as desired.
\end{proof}

\section{Structured Sampling for Imperfect Matchings}
\label{sec:app_imperfect}

In this section, we show that an extension of the Structured Sampling algorithm from Section~\ref{sec:structured-sampling} yields the same competitive ratio for the general case when graph $G$ may not contain a perfect matching.

\begin{lemma}\label{lem:app_skeleton} [Restatement of Lemma 2 from \cite{LeeSingla17}]
Given a bipartite graph $H = (\setpair{U, V_P}, E_H)$, there exists a partition of $V_P$ into sets $S_{-\infty}, S_a$, \ldots, $S_b$ and a partition of $U$ into sets $T_a$, \ldots, $T_b, T_{\infty}$ for some integers $a \leq b$
such that the following hold:
\begin{itemize}
    \item All nodes in $S_{-\infty}$ and $T_\infty$ have degree zero.
    \item $\Gamma(\bigcup_{i < j} S_i) = \bigcup_{i<j} T_i$ for all $j$.
    \item $\frac{|S_{a}|}{|T_{a}|} > \cdots > \frac{|S_{-1}|}{|T_{-1}|} > \frac{|S_0|}{|T_0|} = 1 > \frac{|S_1|}{|T_1|} > \cdots > \frac{|S_b|}{|T_b|}$. \footnote{For reasons of exposition, we allow $S_0 = T_0 = \emptyset$, defining the ratio of their sizes to be 1.}
    \item There is a fractional matching in $H$ such that 
    \begin{itemize}
        \item For all $i$ with $0 \leq i \leq b$, the fractional degree of each node in $S_i$ is 1 and the fractional degree of each node in $T_i$ is $|S_i| / |T_i|$
        \item For all $i$ with $a \leq i < 0$, the fractional degree of each node in $T_i$ is 1 and the fractional degree of each node in $S_i$ is $|T_i| / |S_i|$
        \item For all $i$ with $a \leq i \leq b$, nodes in $S_i$ are only matched with nodes in $T_i$ and vice versa.
    \end{itemize}
\end{itemize}
Further, the $(S_i, T_i)$ pairs can be found in polynomial time.
\end{lemma}

In \cite{goel2012communication} and \cite{LeeSingla17}, the bipartite graph is assumed to have no isolated vertices. In order to handle the case that the predicted graph $H$ has isolated vertices, we extend the decomposition to include the sets $S_{-\infty}$ and $T_{\infty}$ that contain all isolated vertices of $V_P$ and $U$ respectively. Note that we can assume  without loss of generality that $S_{-\infty} = \emptyset$, since any isolated predicted vertex can be dropped from the instance without hurting the algorithm. 

\begin{lemma}
\label{lem:app_matching-extension}
There exists a maximum matching $M^*$ of graph $G$ such that 
\begin{itemize}
    \item $M^*$ restricted to $H$ is a maximum matching of $H$.
    \item $M^*(V_A) \cap T^- = \emptyset$, where $T^- = \bigcup_{i\leq 0} T_i$. 
\end{itemize}
\end{lemma}

\begin{proof}
Let $M^*$ be a maximum matching of $G$ with the maximum number of edges in $H$. Let $\tilde{M}$ be a maximum matching of $H$. Assume, for the sake of contradiction, that $|M^* \cap E(H)| < |\tilde{M}|$. Since $|M^* \cap E(H)| < |\tilde{M}|$, some connected component $P$ of the edge set $M^* \cup \tilde{M}$ has the property that $|P\cap M^* \cap E(H)| < |P\cap \tilde{M}|$. If $P$ is a cycle or even-length path, augmenting the matching $M^*$ with $P$ results in a maximum matching of $G$ with strictly more edges than $M^*$ in $H$, a contradiction. Therefore, $P$ must be an odd-length path.

$P$ must have its endpoint edge(s) in $M^*$, as otherwise augmenting $M^*$ with $P$ would result in a larger matching of $G$. No non-endpoint vertex of $P$ can be in $V_A \setminus V_P$, because each non-endpoint vertex of $P$ is incident with some edge in $H$. Since $P$ has odd length, its endpoints cannot both be in the same bipartition. $P$ cannot entirely be contained in $H$, as otherwise augmenting $\tilde{M}$ with $P$ would result in a larger matching in $H$. Therefore, $P$ has exactly one endpoint vertex $x$ outside of $H$ and the subpath $P\setminus \{x\}$ is entirely contained in $H$. $P\setminus \{x\}$ is a path with even length with $|(P\setminus \{x\})\cap M^*\cap E(H)| = |(P\setminus \{x\})\cap \tilde{M}|$. Therefore, $|P\cap M^*\cap E(H)| = |P\cap \tilde{M}|$, a contradiction. In particular, $|M^*\cap E(H)| = |\tilde{M}|$, so $M^*$ is a maximum matching whose restriction to $H$ is a maximum matching of $H$.

To prove the second part of the claim, we observe that any maximum matching of $H$ must match all vertices in $T^- = \bigcup_{i\leq 0} T_i$.
\end{proof}

Let $M^*$ denote the maximum matching of $G$ that satisfies Lemma \ref{lem:app_matching-extension}. 
Recall from Section \ref{sec:structured-sampling} that our goal is to sample a maximum matching $M$ in $H$ in order to maximize the expected size of $R \cap M^*(V_A)$ where $R \subseteq U$ is the set of offline vertices left unmatched by $M$.  Algorithm \ref{alg:app_structured_preprocessing_imperfect} shows the new preprocessing phase, while the online phase remains the same as earlier (Algorithm \ref{alg:online}). In the preprocessing phase, we  find a decomposition of the predicted graph $H$, and for each index $0 \leq i \leq b$, we sample a random maximum matching between $S_i$ and $T_i$ using Lemma~\ref{lem:rounding}. Similarly, for each index $i < 0$, we find an arbitrary maximum matching.

\begin{algorithm}
\DontPrintSemicolon
\SetAlgoVlined
\newcommand{\forcond}{$i=0$ \KwTo $n$}
\SetKwFunction{FRecurs}{Preprocess}%
\SetKwProg{Fn}{Function:}{\string:}{}
\Fn(){\FRecurs{H}}{
\KwData{Predicted graph $H$}
\KwResult{Maximum matching in $H$, sequence of nodes from $U$}
\BlankLine
Decompose $H$ into $\{(S_i, T_i)\}_{i=a}^b \cup T_{\infty}$ using Lemma~\ref{lem:app_skeleton}\;
For $i < 0$, $M_i \leftarrow $ Arbitrary maximum matching between $S_i$ and $T_i$\;
For $0 \leq i \leq b$, $M_i \leftarrow $ Random matching between $S_i$ and $T_i$ using Lemma~\ref{lem:rounding}\;
$\displaystyle M \leftarrow \bigcup_i M_i$\;
Let $R_{\text{set}} \subseteq U$ be the set of nodes unmatched by $M$\;
$R \leftarrow$ Uniform random permutation of $R_{\text{set}}$\;
\KwRet{M, R}
}
\caption{Structured Sampling: Preprocessing Phase.}
\label{alg:app_structured_preprocessing_imperfect}
\end{algorithm}

\begin{lemma}
\label{lem:app_babylemma}
The matching $M$ found by Algorithm~\ref{alg:app_structured_preprocessing_imperfect} is a maximum matching of $H$.
\end{lemma}

\begin{proof}
By construction, the matching $M$ has size $\sum_{i < 0} |T_i| + \sum_{i \geq 0} |S_i|$. Consider any maximum matching $\tilde{M}$ of $H$. We have $|\tilde{M}| = |\tilde{M}(\bigcup_{i < 0} T_i)| + |\tilde{M}(\bigcup_{i \geq 0} T_i)|$. However, by Lemma~\ref{lem:app_skeleton}, $\Gamma(\bigcup_{i \geq 0} T_i) \subseteq \bigcup_{i \geq 0} S_i$. Thus, $|\tilde{M}(\bigcup_{i \geq 0} T_i)| \leq |\bigcup_{i \geq 0} S_i| = \sum_{i \geq 0} |S_i|$. This yields $|M| \geq |\tilde{M}|$ and the claim follows.
\end{proof}

As in Section~\ref{sec:structured-sampling}, let $n_i = |T_i|$, $d_i = |T_i| - |S_i|$, and $\ell_i = |M^*(V_A) \cap T_i|$. We first present the following technical lemma.

\begin{lemma}\label{lem:app_sequence}
$\sum_{i> 0} \frac{n_i \ell_i}{d_i} \leq \sum_{i> 0} (n_i - d_i + \ell_i)$.
\end{lemma}

\begin{proof}
Recall that, by Lemma \ref{lem:adv_matches}, we have $\sum_{1 \leq i \leq t} \ell_i \leq \sum_{1 \leq i \leq t} d_i$ for all $t$ \footnote{As stated, the lemma only applies to perfect matchings. However, it can be extended to the general case by noting that every node in $S_i$ (for $i\geq 0$) is matched in every maximum matching of $H$.}; note that we removed the $i=0$ terms since $d_0 = \ell_0 = 0$. We first claim that for fixed $d_i$ and $n_i$, the sum $\sum_{i >0} \frac{n_i \ell_i}{d_i}$ is maximized when $\ell_i \leq d_i,\ \forall i > 0$. Suppose for contradiction that $\{\ell_i\}$ is the lexicographically-largest maximum-attaining assignment such that $\ell_j > d_j$ for some $j > 0$. Without loss of generality, let $j$ be the smallest index such that $\ell_j > d_j$. However, in order to satisfy $\sum_{1 \leq i \leq j} \ell_i \leq \sum_{1 \leq i \leq j} d_i$, this implies we must have $\sum_{1 \leq i < j} \ell_i < \sum_{1 \leq i < j} d_i$. Hence, there must be some index $j' < j$ such that $\ell_{j'} < d_{j'}$.

Let $\ell'_i = \ell_i$ for all $i \notin \{j ,j'\}$ and set $\ell'_{j'} = \ell_{j'}+1$ and $\ell'_j = \ell_j - 1$. Notice that we still satisfy $\sum_{1 \leq i \leq t} \ell'_i \leq \sum_{1 \leq i \leq t} d_i$ for all $t$ and $\{\ell'_i\}$ is lexicographically larger than $\{\ell_i\}$. At the same time,
\begin{align*}
    \sum_{i > 0} \frac{n_i \ell'_i}{d_i} = \sum_{i > 0} \frac{n_i \ell_i}{d_i} + \frac{n_{j'}}{d_{j'}} - \frac{n_{j}}{d_{j}} > \sum_{i > 0} \frac{n_i \ell_i}{d_i},
\end{align*}
which is a contradiction. Hence we have proved that $\sum_{i > 0} \frac{n_i \ell_i}{d_i}$ is maximized when $\ell_i \leq d_i,\ \forall i$.

On the other hand, when $\ell_i \leq d_i$, we have $(d_i - \ell_i) \leq (d_i - \ell_i) \frac{n_i}{d_i}$ which implies that $\frac{n_i \ell_i}{d_i} \leq n_i - d_i + \ell_i$. Summing over all $i>0$ completes the proof.
\end{proof}

\begin{theorem}\label{thm:app_marked}
Choose reserved set $R$ according to  Algorithm \ref{alg:app_structured_preprocessing_imperfect}. Then the expected number of marked nodes in $R$ is at least $\delta^2 \nu(G)$. That is, $|R \cap M^*(V_A)| \geq \delta^2 \nu(G)$ in expectation.
\end{theorem}
\begin{proof}
Let $\ell_i = \size{M^*(V_A) \cap T_i}$. By Lemma~\ref{lem:app_matching-extension}, we have $\ell_i = 0$ for all $i < 0$. 
For each $i \geq 0$, every node $u\in T_i$ is chosen to be in $R$ with probability $d_i/n_i$, with the $d_i/n_i$ forming an increasing sequence. So the expected size of $|R \cap M^*(V_A)|$ is given by
\begin{align*}
\mathbb{E}[|R \cap M^*(V_A)|] &= 
\sum_{i \geq 0} \frac{d_i}{n_i}\ell_i
= \sum_{i > 0} \frac{d_i}{n_i}\ell_i.
\intertext{Applying Cauchy--Schwarz with vectors $u$ and $v$ defined as $u_i = \frac{\sqrt{d_i \ell_i}}{\sqrt{n_i}}$ and $v_i = \frac{\sqrt{n_i \ell_i}}{\sqrt{d_i}}$}
	&\geq \frac{(\sum_{i> 0} \ell_i)^2}{\sum_{i>0} \frac{n_i \ell_i}{d_i}}\\
&\geq \frac{(\sum_{i> 0} \ell_i)^2}{\sum_{i>0} (n_i - d_i + \ell_i)} &\mbox{By Lemma \ref{lem:app_sequence}.}
\end{align*}
However by Lemma \ref{lem:app_matching-extension}, $\sum_{i > 0} \ell_i = |M^*(V_A)| = \nu(G) - \nu(H)$. Further, we have  $\nu(G) = \nu(H) + \sum_{i \geq 0} \ell_i \geq \sum_{i>0} (n_i - d_i + \ell_i)$. Substituting these values above and recalling that $\delta = 1 - \frac{\nu(H)}{\nu(G)}$, we get
\begin{align*}
\mathbb{E}[|R \cap M^*(V_A)|] &\geq \frac{(\nu(G) - \nu(H))^2}{\nu(G)} = \delta^2 \nu(G), \mbox{ as desired.}
\end{align*}
\end{proof}

We can now present the main theorem.
\begin{theorem}
\label{thm:app_main-alg}
There is a randomized algorithm for the semi-online bipartite matching problem with a competitive ratio of at least $(1-\delta+\delta^2(1-1/e))$ in expectation, even when the graph does not have a perfect matching.
\end{theorem}

\begin{proof}
By Lemma~\ref{lem:app_babylemma}, the matching $M$ sampled during the preprocessing phase is a maximum matching of $H$ and thus has size $\nu(H)$. Further, by Theorem \ref{thm:app_marked}, we have reserved at least $\delta^2 \nu(G)$ nodes that can be matched to the adversarial nodes. RANKING will match at least a $(1-1/e)$ fraction of these in expectation. So in expectation, the total matching has size at least
$
\nu(H) + \delta^2 \nu(G)(1-1/e) = \nu(G)(1-\delta +\delta^2 (1-1/e))
$ as desired.
\end{proof}

\section{Sampling From Arbitrary Set Systems}
\label{sec:sets-alg}

In Section \ref{sec:structured-sampling}, we used  graph decomposition to sample a matching in the predicted graph such that, in expectation, there is a large overlap between the set of reserved (unmatched) nodes and the unknown set of marked nodes chosen by the adversary. In this section we prove the existence of probability distributions on sets, with this ``large overlap'' property, in settings more general than just bipartite graphs. 

Let $U$ be a universe of $n$ elements and let $\mathcal{S}$ denote a family of subsets of $U$ with equal sizes, i.e., 
$|S| = d, \forall S \in \mathcal{S}$. Suppose an adversary chooses a set $T \in \mathcal{S}$, which is unknown to us. Our goal is to find a probability distribution over $\mathcal{S}$ such that the expected intersection size of $T$ and a set sampled from this distribution is maximized.
We prove in Theorem~\ref{thm:sets-problem} that for any such set system, one can always guarantee that the expected intersection size is at least $\frac{d^2}{n}$. 

The connection to matchings is as follows.
Let $U$, the set of offline nodes in the matching problem, also be the universe of elements. $\mathcal{S}$ is a collection of all maximal subsets $R$ of $U$ such that there is a perfect matching between $U\setminus R$ and $V_P$. All these subsets have size $d=|V_A| = \delta n$. Notice that $M^*(V_A)$ is one of the sets in $\mathcal{S}$, although of course we don't know which. 
What we would like is a distribution such that sampling a set $R$ from it satisfies $\mathbb{E}[|R \cap M^*(V_A)|] \geq d^2/n = \delta^2  n$.

\begin{theorem}
\label{thm:sets-problem}
For any set system $(U, \mathcal{S})$ with $|U|=n$ and $|S| = d$ for all $S \in \mathcal{S}$, there exists a probability distribution $\mathcal{D}$ over $\mathcal{S}$ such that $\forall T\in \mathcal{S}$, $\mathbb{E}_{S \sim \mathcal{D}}[|S \cap T|] \geq \dfrac{d^2}{n}$.
\end{theorem}

As an example, consider $U = \{v, w, x, y, z\}$ and $\mathcal{S} = \{\{v, w\}, \{w,x\}, \{x,y\}, \{y,z\}\}$. Here $n=5$ and $d=2$, so the theorem guarantees a probability distribution on the four sets such that each of them has an expected intersection size with the selected set of at least $\frac{4}{5}$. We can set $\Pr[\{v,w\}] = \Pr[\{y,z\}] =\frac{3}{10}$ and $\Pr[\{w,x\}] = \Pr[\{x,y\}] = \frac{1}{5}$. Then the expected intersection size for the set $\{v,w\}$ is $\Pr[\{v,w\}] \cdot 2 + \Pr[\{w,x\}] \cdot 1 = \frac{4}{5}$ because the intersection size is 2 if $\{v,w\}$ is picked and 1 if $\{w,x\}$ is picked. Similarly, one can verify that the expected intersection for any set is at least $\frac{4}{5}$. However, in general, it is not trivial to find such a distribution via an explicit construction.

Theorem~\ref{thm:sets-problem} is a generalization of Theorem~\ref{thm:marked}, and we could have selected a matching and a reserved set $R$ according to the methods used in its proof. Indeed, this gives the same competitive ratio. However, the set system generated by considering all matchings of size $n-d$ is exponentially large in general. Hence the offline portion of the algorithm would not run in polynomial time.

\subsection{Proof of Theorem~\ref{thm:sets-problem}}
\label{sec:proof-sets-exist}

Let $\calD$ be a probability distribution over $\calS$ with the probability of choosing a set $S$ denoted by $p_S$. Now, for any fixed set $T \in \calS$, the expected intersection size is given by $\mathbb{E}_{S \sim \calD} [|S \cap T|] = \sum_{S \in \calS} p_S \cdot |S \cap T| = \sum_{u \in T} \sum_{S \ni u} p_S$. 
For a given set system $(U, \calS)$, consider the following linear program (Primal-LP) and its dual.

The primal constraints exactly capture the requirement that the expected intersection size is at least $\frac{d^2}{n}$ for any choice of $T$. Thus, to prove the theorem, it suffices to show that the optimal primal solution has an objective value of at most 1. We show that any feasible solution to the dual linear program (Dual-LP) must have objective value at most 1 and hence the theorem follows from strong duality.

\begin{minipage}[t][55mm][t]{0.45\textwidth}
\begin{align}
  &&&\hspace{-1cm}\min \sum_{S \in \calS} p_S \nonumber\\
\text{s.t. }\quad \nonumber \\
  & \forall T \in \calS, && \displaystyle \sum_{u \in T} \sum_{S \ni u} p_S \geq \dfrac{d^2}{n}  \label{eqn:LPassignment}\\
 & \forall S \in \calS, && \displaystyle p_S \geq 0 \label{eqn:LPrelax}
\end{align}
\captionof{figure}{Primal-LP.}
\end{minipage}\hspace{5mm}
\begin{minipage}[t][55mm][t]{0.45\textwidth}
\begin{align}
  &&&\hspace{-1cm}\max \sum_{T \in \calS} q_T \nonumber\\
\text{s.t.}\quad \nonumber \\
  & \forall S \in \calS, && \displaystyle \sum_{u \in S} \sum_{T \ni u} q_T \leq \dfrac{d^2}{n}  \label{eqn:dual-LPassignment}\\
 & \forall T \in \calS, && \displaystyle q_T \geq 0 \label{eqn:dual-LPrelax}
\end{align}
\captionof{figure}{Dual-LP.}
\end{minipage}

\begin{lemma}
\label{lemma:dual-feasibility}
For any set system $(U, \calS)$, the optimal solution to Dual-LP has objective value at most 1.
\end{lemma}

\begin{proof}
Let $\{q_T\}_{T \in \calS}$ denote an optimal, feasible solution to Dual-LP. For any element $u \in U$, define $w(u) = \sum_{T \ni u} q_T$ to be the total weight of all the sets that contain $u$. From the dual constraints, we have

\begin{align}
\quad\quad\quad\quad\quad\quad\quad\quad\quad\quad\quad &\forall S \in \calS, \quad \sum_{u \in S} w(u) \leq \dfrac{d^2}{n}. \nonumber
\intertext{Since each $S \in \calS$ has exactly $d$ elements, we can rewrite the above as}
&\forall S \in \calS, \quad \sum_{u \in S} \left(w(u) - \dfrac{d}{n}\right) \leq 0. \nonumber
\end{align}
Multiplying each inequality by $q_S$ and adding over all $S \in \calS$ yields
\begin{eqnarray*} 
0 & \geq & \sum_{S \in \calS} \left(\sum_{u \in S} \left(w(u) - \dfrac{d}{n}\right)q_S\right) \\
& = & \sum_{u \in U} \left(\sum_{S \ni u} \left(w(u) - \dfrac{d}{n}\right)q_S\right) \\
& = & \sum_{u \in U} w(u) \left(w(u) - \dfrac{d}{n}\right) \\
& \geq & \sum_{u \in U} \frac{d}{n} \left(w(u) - \dfrac{d}{n}\right),
\end{eqnarray*}
where we used that $y(x-y) \leq x(x-y)$ for any two real numbers $x$ and $y$.  Hence we obtain
\begin{equation}
\sum_{u \in U} w(u) \leq d. \label{ineq:final}
\end{equation}
On the other hand, we have 
\begin{equation}
\sum_{u \in U} w(u) = \sum_{u \in U} \sum_{T \ni u} q_T = \sum_{T \in \calS} \sum_{u \in T} q_T = d \sum_{T \in \calS} q_T.  \nonumber
\end{equation}
Inequality~(\ref{ineq:final}) then shows that  $\sum_{T \in \calS} q_T \leq 1$.
\end{proof}

\newcommand{\ALG}{\mathrm{ALG}}
\newcommand{\OPT}{\mathrm{OPT}}

\section{Fractional Matching}

In this section we consider algorithms for the semi-online \emph{fractional} bipartite matching problem and give tight approximation results using the primal-dual technique. The hardness result of~\cite[Theorem 3.23]{esfandiari2018} applies to the fractional version as well, showing that no approximation better than $1 - \delta e^{-\delta}$ is possible.

We consider the same model for arrival of nodes and edges as described in Section \ref{sec:model}. However, instead of generating an integral matching, the algorithm is required to construct a fractional matching. In particular, when each node $v\in V$ is revealed, the algorithm has to assign fractional values $x_{\{u,v\}} \geq 0$ to edges of $v$ such that $\sum_{w \in \delta(z)} x_{\{w,z\}} \leq 1$ for every vertex $z\in U \cup V$.

Section \ref{sec:frac-pd} presents the primal-dual analysis showing an approximation guarantee for the special case in which the optimal solution fully matches all the offline nodes $U$ in the instance. Section \ref{sec:frac-gen} shows that the same guarantee also holds for arbitrary bipartite graphs, by using a monotonicity property of the algorithm.

We use $\ALG(G)$ to denote the value of the solution found by our algorithm on an instance $G$, and $\OPT(G)$ to denote the value of an optimal solution. Note that, by integrality of maximum matchings, we can assume that the optimal solution to an instance of the fractional matching problem assigns all edges to an extent of $0$ or $1$. Thus, we still use $M^*$ to denote the optimal matching.

\subsection{Case when $U$ is fully matched in OPT} \label{sec:frac-pd}
For simplicity, we first consider the special case when the optimal solution fully matches all the offline vertices in $U$, i.e. $|M^*| = |U|$.
In this section, we prove the following theorem.

\begin{theorem}\label{thm:frac-full}
There is a deterministic algorithm for the semi-online fractional bipartite matching problem with a competitive ratio of $(1 - \delta e^{-\delta})$ when the input graph $G$ has a matching that matches all vertices in $U$.
\end{theorem}

\paragraph*{Algorithm} 
Similarly to the algorithm in Section \ref{sec:structured-sampling}, we first find a skeleton decomposition of the bipartite graph $H = (U, V_P, E_H)$ using Lemma \ref{lem:app_skeleton}. Further, we use the fractional matching guaranteed by Lemma \ref{lem:app_skeleton} as the maximum matching on $H$. For any vertex $u \in U$, let $y_u$ denote the fractional degree of $u$ in the obtained fractional matching. 

We first note that we can assume for simplicity that all vertices in $V_P$ arrive before any vertex from $V_A$ arrives. Indeed, if this is not the case, since the algorithm knows the identity of vertices in $V_P$, we can simply pretend that all vertices in $V_P$ arrive first and pretend that they get matched as per the fractional matching obtained using the skeleton decomposition. Once these vertices actually arrive in the online order, they can be assigned using the matching already found.
Finally, in the online phase, we use the standard water level algorithm~\cite{kalyanasundaram2000balance} to find a fractional matching for all the vertices in $V_A$. More formally, each online vertex $v$ is matched at a uniform rate to all of its neighbors with the least fractional degree until it is fully matched or all of its neighbors are fully matched.


\paragraph*{Analysis} We analyze the above algorithm using the primal-dual framework~\cite{devanur2013randomized}.
Figures \ref{fig:fractional-primal} and \ref{fig:fractional-dual} show the primal and dual linear programs respectively for the fractional matching problem.

\begin{minipage}[t][60mm][t]{0.45\textwidth}
\begin{align}
  &&&\hspace{-1cm}\max \sum_{e \in E} x_e \nonumber\\
\text{s.t. }\quad \nonumber \\
  & \forall u \in U, && \displaystyle \sum_{v \in \delta(u)} x_{\{u,v\}} \leq 1  \label{eqn:LPdegree_u}\\
  & \forall v \in V, && \displaystyle \sum_{u \in \delta(v)} x_{\{u,v\}} \leq 1  \label{eqn:LPdegree_v}\\
 & \forall E \in E, && \displaystyle x_e \geq 0 \label{eqn:LP_relax}
\end{align}
\captionof{figure}{Primal-LP.}
\label{fig:fractional-primal}
\end{minipage}\hspace{5mm}
\begin{minipage}[t][60mm][t]{0.45\textwidth}
\begin{align}
  &&&\hspace{-1cm}\min \sum_{u \in U} \alpha_u + \sum_{v \in V} \beta_v \nonumber\\
\text{s.t.}\quad \nonumber \\
  & \forall \{u,v\} \in E, && \displaystyle \alpha_u + \beta_v \geq 1 \label{eqn:dual-edge-cover}\\
 & \forall u \in U, && \displaystyle \alpha_u \geq 0 \label{eqn:alpha-non-neg}\\
 & \forall v \in V, && \displaystyle \beta_v \geq 0 \label{eqn:beta-non-neg}
\end{align}
\vspace{7.5mm}
\captionof{figure}{Dual-LP.}
\label{fig:fractional-dual}
\end{minipage}

Let $\{\tilde{x}_e\}_{e \in E}$ denote the fractional matching obtained by the semi-online algorithm. By construction, we are guaranteed that the algorithm always constructs a feasible fractional matching.
We'll use a dual-fitting argument to set dual variables $\{\alpha_u\}_{u \in U}$ and $\{\beta_v\}_{v \in V}$ so that the following two conditions hold true.
\begin{enumerate}
    \item \label{prop1} The primal objective is at least the dual objective, i.e., $\sum_{e \in E} \tilde{x}_e \geq \sum_{u \in U} \alpha_u + \sum_{v \in V} \beta_v$.
    \item \label{prop2} The dual solution is almost feasible, i.e. for any edge $(u,v)$, $\alpha_u + \beta_v \geq 1 - \delta e^{-\delta}$.
\end{enumerate}

We observe that these two properties are sufficient to prove Theorem \ref{thm:frac-full}. This is because scaling up all dual variables by $\left(\frac{1}{1-\delta e^{-\delta}}\right)$ yields a feasible dual solution and thus by weak duality and condition~\ref{prop1}, 
\begin{eqnarray*}
\OPT(G) & \leq & \left(\frac{1}{1-\delta e^{-\delta}}\right)\left(\sum_{u \in U} \alpha_u + \sum_{v \in V} \beta_v\right) \\
& \leq & \left(\frac{1}{1-\delta e^{-\delta}}\right) \sum_{e \in E} \tilde{x}_e \\
& = & \left(\frac{1}{1-\delta e^{-\delta}}\right) \ALG(G),
\end{eqnarray*}
and the theorem follows.

After the preprocessing phase, we initialize the dual variables as follows. 
For each offline vertex $u \in T_i$, set $\alpha_u = e^{-\delta_i} - \delta e^{-\delta}$ where $\delta_i = 1 - \frac{|S_i|}{|T_i|}$. For each predicted vertex $v \in S_i$, set $\beta_v = 1 - e^{-\delta_i}$. (We note that since this is just the analysis and not part of the algorithm, we can use $\delta$ even if the algorithm doesn't know its value.)
The dual variables are updated in the online phase as follows. Whenever the primal $\tilde{x}_{\{u,v\}}$ increases by $dx$ for any edge $(u,v)$, also increase the corresponding dual variables: $\alpha_u$ by $e^{y_u - 1} dx$ and $\beta_v$ by $(1 - e^{y_u - 1}) dx$ where $y_u$ denotes the instantaneous fractional degree of vertex $u$.

\begin{lemma}
$\alpha_u \geq 0$ and $\beta_v \geq 0$ for all $u\in U$ and $v\in V_P$.
\end{lemma}
\begin{proof}
We verify the claim after initialization; subsequently, these values can only increase.
For $\delta_i \in [0, 1]$, $e^{-\delta_i}$ is minimized at $\delta_i=1$. For $\delta \in [0, 1]$, $(-\delta e^{-\delta})$ is minimized at $\delta=1$. Thus, the initial value of $\alpha_u \ge e^{-1} - 1 e^{-1} = 0$.
Similarly, $\beta_v$ is minimized at $\delta_i=1$, so $\beta_v \geq 1-e^0 = 0$.
\end{proof}

We now verify conditions \ref{prop1} and \ref{prop2} above.

\begin{lemma}
Condition \ref{prop2} is satisfied for the predicted edges $E_H$.
\end{lemma}
\begin{proof}
Consider an edge $uv$ and suppose that $u\in T_i$ and $v\in S_j$. By the properties of the skeleton decomposition, it must be that $i\leq j$. This implies that $\delta_i = 1- \frac{S_i}{T_i} \leq 1- \frac{S_j}{T_j} = \delta_j$.
After the initialization of the dual variables, $\alpha_u + \beta_v = e^{-\delta_i} - \delta e^{-\delta} + 1 - e^{-\delta_j} \geq 1- \delta e^{-\delta}$. During the online phase, $\alpha_u$ can only increase.
\end{proof}

\begin{lemma}
At the end of the algorithm, condition \ref{prop2} is satisfied for edges in $E_G \setminus E_H$.
\end{lemma}
\begin{proof}
Consider an arbitrary online edge $(u,v)$. We have two cases, either $u$ is full (has fractional degree 1) or it is not.
If $u$ is full, then we have $\alpha_u = e^{-\delta_i} - \delta e^{-\delta} + \int_{1-\delta_i}^1 e^{x-1} dx = 1 - \delta e^{-\delta}$. On the other hand, if $u$ is not full, let $y_u$ denote its final fractional degree. Then $\alpha_u = e^{-\delta_i} - \delta e^{-\delta} + \int_{1-\delta_i}^{y_u} e^{x-1} dx = e^{y_u - 1} - \delta e^{-\delta}$. But in this case, notice that $v$ must be fully allocated and further should be allocated only to offline vertices whose fractional degree (at the time) was less than $y_u$. Thus, $\beta_v \geq 1 - e^{y_u - 1}$. Thus, $\alpha_u + \beta_v \geq 1 - \delta e^{-\delta}$ as desired.
\end{proof}

\begin{lemma}
Condition \ref{prop1} holds throughout the algorithm.
\end{lemma}
\begin{proof}
During the online phase, each increase of $dx$ in the primal is accompanied by an equal increase of $e^{y_u - 1} dx + (1 - e^{y_u - 1}) dx$  in the sum of the dual variables, maintaining the inequality.

Now, the only part of the argument required to complete the proof is to show that condition \ref{prop1} is satisfied by the initial allocation of the dual values. We need to show that
\begin{align*}
    \text{total initial dual} &\leq \text{total initial primal (matching)}\\
    \Leftrightarrow \sum_{i} \left[ n_i \left( e^{-\delta_i} - \delta  e^{-\delta}\right) + (n_i - d_i) (1 - e^{-\delta_i})\right] &\leq \sum_i (n_i - d_i)\\
    \Leftrightarrow \sum_{i} d_i e^{-\delta_i} &\leq d e^{-\delta}\\
    \Leftrightarrow \sum_{i} \dfrac{d_i}{n} e^{-\delta_i} &\leq \delta e^{-\delta}.
    \intertext{But the above inequality is true by Jensen's (or definition of concavity). Consider}
    \mathrm{LHS} &= \sum_{i} \dfrac{d_i}{n} e^{-\delta_i}\\
    &= \sum_{i} \dfrac{n_i}{n} \delta_i e^{-\delta_i}\\
    \intertext{Consider $f(x) = x e^{-x}$. $f$ is concave on the interval $(-\infty, 2)$, which includes the possible values of $\delta_i \in [0, 1]$; so applying $\sum_i \alpha_i f(x_i) \leq f(\sum_i \alpha_i x_i)$, we get}
    &\leq f(\sum_i \dfrac{n_i}{n} \delta_i) = f(\delta)\\
    &= \mathrm{RHS}.
\end{align*}
This completes the proofs of the lemma and the theorem.
\end{proof}

\subsection{Fractional matching for arbitrary bipartite graphs} \label{sec:frac-gen}

\paragraph*{Algorithm} 
Let ALG be the following algorithm. In the preprocessing phase, find the skeleton decomposition and assign all predicted online vertices fractionally so that $y_u = \frac{n_i - d_i}{n_i} = 1 - \delta_i$ for all $u \in T_i$ where $y_u$ is the fractional degree of vertex $u$. In the online phase, use the water-level algorithm.

\paragraph*{Analysis} 

Among integral optimal matchings for $G$, let $M^*$ be one that maximizes the matching size in $H$. 
We define $U' \subseteq U$ to be the set of offline nodes matched by $M^*$, and $G'$ be the instance just like $G$, except with nodes in $U\setminus U'$ and their incident edges removed.

\begin{lemma}
$\OPT(G') = \OPT(G)$.
\end{lemma}
\begin{proof}
$M^*$ does not match the nodes $U\setminus U'$ that were removed, so the same solution is feasible for $G'$.
\end{proof}

\begin{lemma}\label{lem:size_h}
The size of a maximum matching on $H$ is equal to that on $H'$.
\end{lemma}
\begin{proof}
 By Lemma \ref{lem:app_matching-extension} and the choice of $M^*$, $M^*$ restricted to $H$ is a maximum matching in $H$. Since all endpoints of $M^*$ are preserved in $H'$, the size of the maximum matching is unaffected.
\end{proof}

The following lemma is proved in section \ref{sec:proof-alg-monot}.
\begin{lemma}\label{lem:alg-monot}
$\ALG(G) \geq \ALG(G')$.
\end{lemma}

\begin{theorem}
$\ALG$ is a $(1 - \delta e^{-\delta})$-approximation for the semi-online fractional matching problem.
\end{theorem}
\begin{proof}
For a given instance $G$, we define the instance $G'$ as above. ($G'$ is only used for analysis and is not needed by the algorithm.) Then by the preceding lemmas and Theorem \ref{thm:frac-full} applied to $G'$, we have
$$\ALG(G) \geq \ALG(G') \geq (1 - \delta e^{-\delta}) \cdot \OPT(G') = (1 - \delta e^{-\delta}) \cdot \OPT(G).$$
\end{proof}

\subsubsection{Proof of Lemma \ref{lem:alg-monot}} \label{sec:proof-alg-monot}
We need to show that when our algorithm ALG runs on an instance $G$, which has a superset of offline nodes compared to $G'$, then it finds a fractional matching that is at least as big. To do this, we compare the runs of ALG on $G$ and $G'$ step by step.
Let us say that the \emph{water level} of an offline node $u$ at a particular point in the algorithm is the extent to which it has been matched to online nodes up to that point. Recall that ALG fractionally matches $u$ in the preprocessing step, and then possibly increases its water level some more in the online phase. We consider the following key property.


\begin{equation}\tag{$*$}\label{prop}
    \text{Water level of each node $u\in U'$ is lower for $G$ than it is for $G'$}
\end{equation}

\noindent Then, in a sequence of proofs, we develop the following argument.

\begin{enumerate}
    \item The sizes of fractional matchings found by the preprocessing phase in $G$ and $G'$ are equal (Lemma \ref{lem:init-eq}).
    \item Property (\ref{prop}) holds after the preprocessing phase (Corollary \ref{cor:init-level}).
    \item If (\ref{prop}) holds at the beginning of a step of the water-filling algorithm in which a node $v\in V_P$ is processed, then:
    \begin{enumerate}
        \item The extent to which $v$ is matched in $G$ is the same or higher than it is in $G'$ (Lemma~\ref{lem:match-more}).
        \item (\ref{prop}) still holds at the end of this step (Lemma \ref{lem:level-ind}).
    \end{enumerate}
\end{enumerate}

Given these results, Lemma \ref{lem:alg-monot} easily follows: the sizes of the two fractional matchings are the same after the preprocessing phase, and the one in $G$ increases at least as much as in $G'$ with each step of the online phase.

The rest of the section proves the above claims.

\begin{lemma}\label{lem:init-eq}
The sizes of fractional matchings found by the preprocessing phase in $G$ and $G'$ are equal.
\end{lemma}
\begin{proof}
In the preprocessing phase, ALG selects a fractional matching whose size is equal to the maximum matching on $H$. 
By Lemma \ref{lem:size_h}, the maximum matching on $H$ is equal to that on $H'$. Thus, the preprocessing phase of ALG finds equal-sized matchings in both cases.
\end{proof}

\begin{lemma}\label{lem:decomp-mon}
Consider the bipartite graphs $H=(U, V_P, E)$ and $H'=(U', V_P, E')$ where $U'\subseteq U$ as described above. 
Let $(\{S_i\}, \{T_i\})$ be the skeleton decomposition of $H$ and $(\{S'_i\}, \{T'_i\})$ be the skeleton decomposition of $H'$.
For any node $u \in U'$, if $u$ is in the components $T_i$ and $T'_{i'}$ of the two decompositions respectively, then
$\frac{|S_i|}{|T_i|} \leq \frac{|S'_{i'}|}{|T'_{i'}|}$.
\end{lemma}
\begin{proof}
Consider the fractional matching guaranteed by Lemma \ref{lem:app_skeleton} on the graph $H$ and view it as a flow directed from nodes in $V_P$ to nodes in $U$. Let $f$ be that flow, and $f'$ be the analogous one for $H'$.
For convenience, we let $f'$ be defined on the whole graph $H$, with edges in $E\setminus E'$ having zero flow. The total amounts of flow in $f$ and $f'$ are equal to the size of the maximum matchings in $H$ and $H'$ respectively, which are equal for the two graphs by Lemma \ref{lem:size_h}.

Define $f(u) = \sum_{v} f(v,u)$ for $u\in U$ to be the total flow to node $u$, and similarly  $f(v)$ for $v\in V_P$, $f'(u)$, and $f'(v)$. Let $\ell(u)$ be the index of the component in the skeleton decomposition of $H$ that contains $u$, i.e.\ $\ell(u) = i$ if $u\in S_i \cup T_i$. Then $\ell'(u)$ is defined similarly for $H'$. With this notation, the lemma states that for all $u\in U'$, $f(u) \leq f'(u)$.

We consider the difference between the two flows, $h = f - f'$. In particular, for an edge $(v,u)$, if $f(v, u) \geq f'(v, u)$, then $h(v, u) = f(v, u) - f'(v, u)$. If $f(v, u) < f'(v, u)$, then $h$ has flow in the opposite direction: $h(u, v) = |f(v, u) - f'(v, u)|$.
Thus, $f'$ can be transformed into $f$ by adding $h$ to it. 
The flow $h$ can be decomposed into a collection of cycles and maximal paths, such that each path in the decomposition starts at a node $w\in U'$ with $f(w) < f'(w)$ and ends at some $u \in U$ with $f(u) > f'(u)$. Note that these paths do not start or end on the $V_P$ side, as for all $v\in V_P$, $f(v)=f'(v)$.

Suppose for contradiction that there is a node $u\in U'$ such that $f(u) > f'(u)$. Then there must be a directed path in the decomposition of $h$ from some node $w$ to $u$. We consider here the case that this path consists of just two edges, $(w,v)$ and $(v,u)$, for some $v\in V_P$, but the same argument can be applied inductively for the case of a longer path. We make a series of claims that lead to a contradiction, thus proving the lemma.

\begin{enumerate}
    \item \label{ineq1}
    $f(u) > f'(u)$: by assumption.
    \item \label{ineq2}
$f'(w) > f(w)$: because there is a path in the decomposition of $h$ that starts at $w$.
    \item \label{ineq3}
$\ell'(v) = \ell'(w)$: the fact that $h$ has an edge $(w,v)$ means that $f'(v,w)>0$, which means that $v$ and $w$ are in the same component of the decomposition of $H'$.
    \item \label{ineq4}
$\ell(v) = \ell(u)$: the fact that $h$ has an edge $(v,u)$ means that $f(v,u)>0$, which means that $v$ and $u$ are in the same component of the decomposition of $H$.
    \item \label{ineq5}
$\ell'(v) \geq \ell'(u)$: follows from the fact that there is an edge $(v,u)$ in $H'$ by the property of the skeleton decomposition that there are no edges from any $S_i$ to $T_j$ with $i<j$.
    \item \label{ineq6}
$\ell(v) \geq \ell(w)$: follows from the fact that there is an edge $(v,w)$ in $H$ by the property of the skeleton decomposition that there are no edges from any $S_i$ to $T_j$ with $i<j$.
    \item \label{ineq7}
$\ell'(w) \geq \ell'(u)$: combining inequalities \ref{ineq3} and \ref{ineq5}.
    \item \label{ineq8}
$\ell(u) \geq \ell(w)$: combining inequalities \ref{ineq4} and \ref{ineq6}.
    \item \label{ineq9}
    $f'(u) \geq f'(w)$: from \ref{ineq7} and the fact that nodes $x \in T'_i$ in sets with lower index have higher $f'(x)$.
    \item \label{ineq10}
    $f(w) \geq f(u)$: from \ref{ineq8} and the fact that nodes $x \in T_i$ in sets with lower index have higher $f(x)$.
    \item \label{ineq11}
    $f'(u) \geq f'(w) > f(w) \geq f(u)$: combining inequalities \ref{ineq9}, \ref{ineq2}, \ref{ineq10}.
\end{enumerate}

\noindent The last inequality implies that $f'(u) > f(u)$, which contradicts our assumption \ref{ineq1}.
\end{proof}

\begin{corollary}\label{cor:init-level}
After the preprocessing phase of ALG, the water level of a node $u\in U'$ is lower for $G$ than for $G'$.
\end{corollary}
\begin{proof}
Reusing the notation of Lemma \ref{lem:decomp-mon}, after the preprocessing phase of ALG, the water level of $u$ in $G$ is $\frac{|S_i|}{|T_i|}$, and that in $G'$ is $\frac{|S'_{i'}|}{|T'_{i'}|}$. The result follows from the lemma.
\end{proof}

\begin{lemma}\label{lem:match-more}
If (\ref{prop}) holds at the beginning of a step of the water-filling algorithm in which a node $v\in V_P$ is processed, then the extent to which $v$ is matched in $G$ is the same or higher than it is in $G'$.
\end{lemma}
\begin{proof}
The neighbors of $v$ in $G$ consist of nodes in $U'$, whose water level is lower than in $G'$, as well as possibly nodes in $U\setminus U'$, which are not present in $G'$. In the water-filling algorithm, $v$'s matching is increased until either it is fully matched or all its neighbors are fully matched. In either case, with more available capacity of the neighbors, $v$ will be matched to the same or higher extent in $G$.
\end{proof}

\begin{lemma}\label{lem:level-ind}
If (\ref{prop}) holds at the beginning of a step of the water-filling algorithm in which a node $v\in V_P$ is processed, then (\ref{prop}) still holds at the end of this step.
\end{lemma}
\begin{proof}
The water level of nodes that are not neighbors of $v$ doesn't change, so we only consider the neighbors of $v$.
Consider the case that in $G'$, the water-filling step ends because all of $v$'s neighbors become full. Then their final water level is $1$, and the lemma follows as the final level in $G$ can't exceed $1$.

The remaining case is that $v$ is fully matched in $G'$. By Lemma \ref{lem:match-more}, $v$ must be fully matched in $G$ as well.
In the water-filling algorithm, all offline nodes whose level changes in a particular step have the same level at the end of the step. For $G$, let $z$ be this common level;  $y_u$ be the initial water level of a node $u$, and $\Gamma(v)$ be the neighbors of $v$. For $G'$, we define $z'$, $y'_u$, and $\Gamma'(v)$ analogously.

We'd like to show that $z \leq z'$, so assume for contradiction that $z > z'$. By assumption that (\ref{prop}) holds at the beginning of the step, we have $y_u \leq y'_u$ for all $u\in U'$.
For a node $u \in \Gamma(v)$, the water level increases in this step if $y_u<z$, in which case it increases by $(z-y_u)$.
As $v$ is fully matched in both $G$ and $G'$, the total increase in both cases is equal to $1$. So
$$
1 =
\sum_{\substack{u\in \Gamma(v) \\ y_u<z}} (z - y_u) \geq
\sum_{\substack{u\in \Gamma'(v) \\ y_u<z}} (z - y_u) \geq
\sum_{\substack{u\in \Gamma'(v) \\ y'_u<z'}} (z - y_u) >
\sum_{\substack{u\in \Gamma'(v) \\ y'_u<z'}} (z' - y'_u) = 
1,
$$
where the middle inequality follows because the nodes satisfying $y'_u<z'$ are a subset of those satisfying $y_u<z$, due to our assumptions. Having arrived at a contradiction, we conclude that  $z \leq z'$.

Now, the final water level of each neighbor of $v$ at the end of the step is $\max(y_u, z)$ in $G$ and $\max(y'_u, z')$ in $G'$. As $y_u \leq y'_u$ and $z \leq z'$, the lemma follows.
\end{proof}

\section{Extensions - Imperfect Predictions and Agnosticism}
\label{sec:extensions}

In this section, we consider a more general model where we allow the predicted graph to have small random errors. We define the $(d, \epsilon)$ semi-online model as follows - 
We are given a predicted graph $H = (\setpair{U, V}, E_H)$, where $|U| = |V| = n$. As before, $U$ are the offline nodes and $V$ are the online\footnote{The algorithm does not know the arrival order of nodes in $V$.} nodes. However, we do not explicitly separate $V$ into predicted and adversarial nodes; all nodes are seen by the offline preprocessing stage, but some subset of these nodes will be altered adversarially.

An adversary selects up to $d$ online nodes and may arbitrarily change their neighborhoods. In addition, we allow the realized graph $G$ to introduce small random changes to the remaining predicted graph after the adversary has made its choices. Specifically, each edge in $H$ not controlled by the adversary is removed independently with probability $\epsilon$. Further, for each $u\in U, v\in V$, we add edge $(u,v)$ (if it does not already exist in the graph) independently with probability $\epsilon |M|/n^2$, where $M$ is a maximum matching in $H$. Note that in expectation, we will add fewer than $\epsilon |M|$ edges; simply adding edges with probability $\epsilon$ (instead of $\epsilon |M|/n^2$) would overwhelm the embedded matching.
We call an algorithm \emph{agnostic} if it does not know the $d$ nodes chosen by the adversary during the preprocessing (offline) phase. There are two variants - either the algorithm knows the value of $d$ or it does not. We show a hardness result in the former case and consider algorithms in the latter case.

We first consider agnostic algorithms to find integral matchings in this $(d,\epsilon)$ semi-online model and give a hardness result and a corresponding tight algorithm for the case when $\epsilon = 0$.

\begin{theorem}\label{thm:agnostic hard}
In the $(d,\epsilon)$ semi-online model with $d<n/4$,
no (randomized) agnostic algorithm can find a matching of size more than $n - d - \epsilon (n-3d) + O(\epsilon^2 n)$ in expectation, taken over the randomness of the algorithm and the randomness of the realized graph. This holds even if $d$ is known in advance by the algorithm.
\end{theorem}

\begin{proof}
Assume $n$ is even.
Our hard instance consists of the following predicted graph $H$: 
For each integer $i \in  [0,\frac{n}{2})$, add edges $(v_{2i + 1}, u_{2i+1})$, $(v_{2i + 1}, u_{2i+2})$, and $(v_{2i+2}, u_{2i+2})$. This creates $n/2$ connected components. See Figure \ref{fig:agnostic-hardness} for an illustration.

The adversary chooses $d$ components uniformly at random. Let $\mathcal{A} = \{i_1, i_2, \ldots, i_d\} \subset [0, \frac{n}{2})$ denote the indices of the $d$ components selected by the adversary. For each index $i \in \mathcal{A}$, the adversary then selects $v_{2i+2}$ and changes its neighborhood so it only connects with $u_{2i+1}$ (instead of $u_{2i+2}$).

\begin{figure}[htbp]
\centering
\includegraphics[width=0.7\textwidth]{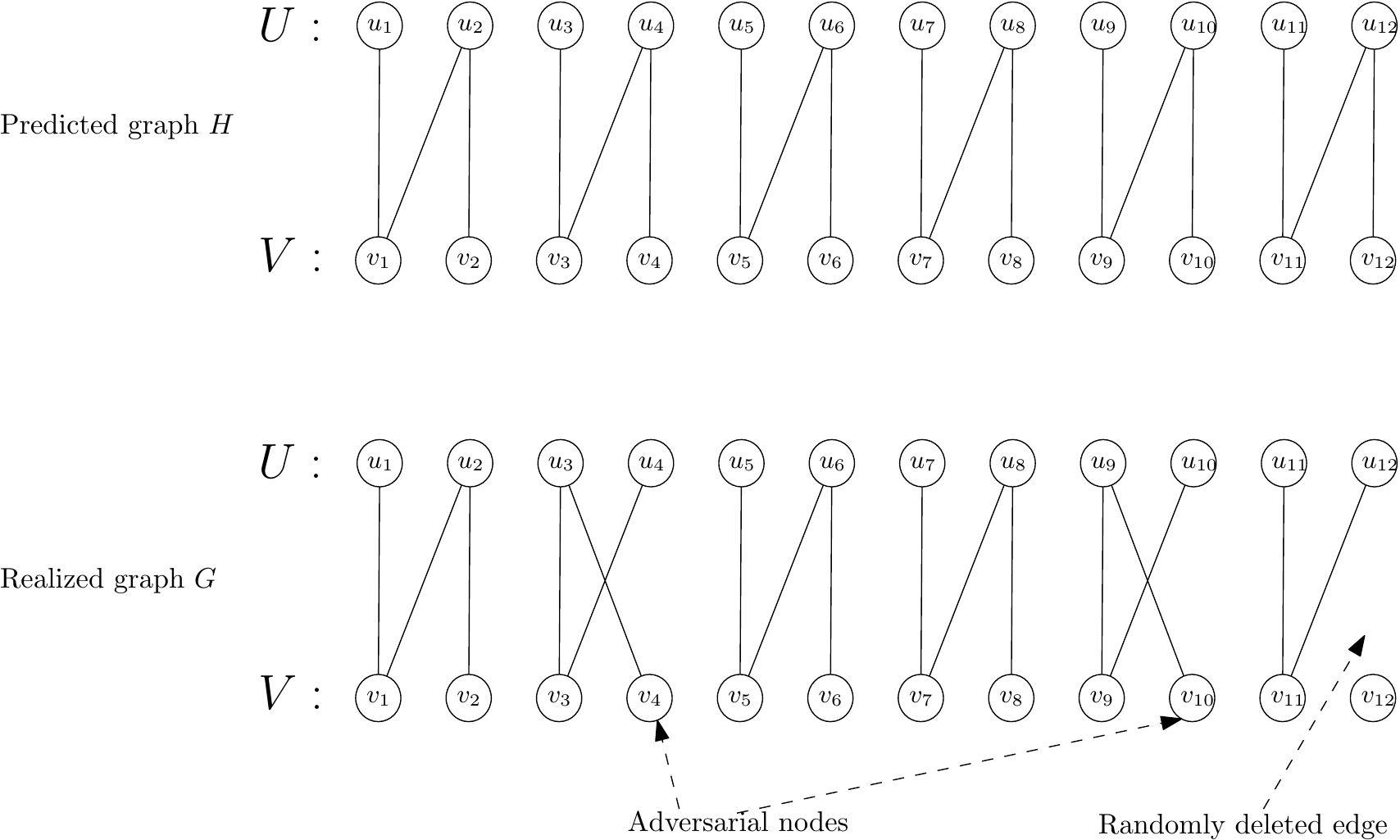}
\caption{Hard instance for Agnostic algorithms}
\label{fig:agnostic-hardness}
\end{figure}

For simplicity, let's first consider the case when $\epsilon = 0$. The algorithm can do no better than picking some $p\in[0,1]$, and matching $v_{2i+1}$ to $u_{2i+1}$ with probability $p$, and matching $v_{2i+1}$ to $u_{2i+2}$ otherwise, for all $i$. The algorithm then matches $v_{2i+2}$ to its neighbor, if possible.
Now, for all $i \in \mathcal{A}$ (components selected by the adversary), this gets an expected matching of size $p + 2(1-p) = 2 - p$. On the other hand, for all $i \notin \mathcal{A}$, the expected matching is size $2p + (1-p) = 1+p$. Since there are $d$ components with an adversary and $n/2 - d$ components without, this gives a total matching of size $(2-p)d + (1+p)(n/2-d)=n/2+d+p(\frac{n}{2}-2d)$. This is maximized when $p=1$ (since $d<n/4$) to yield a matching of size $n-d$.

When $\epsilon > 0$, the algorithm still should set $p = 1$; if the desired edge is removed, then the algorithm will match with whatever node is available. Components with an adversarial node \emph{gain} an edge in the matching when the edge $(v_{2i+1}, u_{2i+1})$ is removed since the algorithm is forced into the right choice; if both edges $(v_{2i+1}, u_{2i+1})$ and $(v_{2i+1}, u_{2i+2})$ are removed, we neither gain nor lose. The expected gain is $\epsilon-\epsilon^2$.
Components without an adversarial node lose an edge in the matching whenever either edge $(v_{2i+1},u_{2i+1})$ or edge $(v_{2i+2},u_{2i+2})$ is removed, and they lose an additional edge if all three edges of the component are removed. So the expected loss is $2\epsilon - \epsilon^2 + \epsilon^3$
Since there are $d$ components with adversarial nodes and $n/2-d$ without, this is a total of loss of
$$
-d(\epsilon-\epsilon^2) + (n/2-d)(2\epsilon - \epsilon^2 + \epsilon^3)
	= \epsilon(n-3d) - O(\epsilon^2 n)
$$
Hence, the total matching is size
$n - \epsilon(n-3d) + O(\epsilon^2 n)$, as claimed.
\end{proof}

\begin{theorem}
Given a predicted graph $H$ with a perfect matching, suppose there are $d$ adversarial nodes and $\epsilon = 0$ as described above in the $(d,\epsilon)$ semi-online model. Then there is an agnostic algorithm that does not know $d$ that finds a matching of expected size $n - d$.
\end{theorem}
\begin{proof}
Before any online nodes arrive, find a perfect matching $M$ in $H$. In the online stage, as each node $v$ arrives,
we attempt to identify $v$ with an online node in the predicted graph with the same neighborhood, and match $v$ according to $M$. If no node in the predicted graph has neighborhood identical to $v$, we know that $v$ is adversarial and we can simply leave it unmatched.
(Note that adversarial nodes can mimic non-adversarial nodes, but it doesn't actually hurt us since they are isomorphic.) The predicted matching had size $n$, and we lose one edge for each adversarial node, so the obtained matching has size $n-d$.
\end{proof}

\subsection{Fractional matchings for predictions with errors}
In this section, we show that we can find an almost optimal \emph{fractional} matching for the $(d, \epsilon)$  semi-online matching problem. 

We use a result from~\cite{vee2010matching}, which gives a method of reconstructing a fractional matching using only the local structure of the graph and a single stored value for each offline node.
They provide the notion of a \emph{reconstruction function}. Their results extend to a variety of linear constraints and convex objectives, but here we need only a simple reconstruction function. For any positive integer $k$, define $g^k:(\mathtt{R}^+_0)^k \rightarrow (\mathtt{R}^+_0)^k$ by
 \[g^k(\alpha_1, \alpha_2, \ldots, \alpha_k) = (\alpha_1 - \max(0, z),~ \alpha_2 - \max(0, z),~ \ldots,~ \alpha_k - \max(0, z))\]
where $z$ is a solution to $\sum_j \min(\max(0, \alpha_j - z), 1) = 1$.

The reconstruction function $g$ is this family of functions. Note that this is well-defined: there is always a solution $z$ between $-1$ and the largest $\alpha_j$, and the solution is unique unless $z\leq 0$.

Lemma~\ref{lem:vee} assigns a value $\alpha_u$ to each $u$ in the set of offline nodes, and reconstructs a matching on the fly as each online node arrives, using only the neighborhood of the online node and the stored $\alpha$ values. Crucially, the reconstruction assigns reasonable values even when the neighborhood is different than predicted. In this way, it is robust to small changes in the graph structure.

\begin{lemma}[Restated from \cite{vee2010matching}]
\label{lem:vee}
Let $g^k_i$ be defined as above, and 
let $H = (\setpair{U, V}, E_H)$ be a bipartite graph with a perfect matching of size $n=|U|=|V|$. Then there exist values $\alpha_u$ for each $u\in U$ (which can be found in polynomial time)
such that the following holds: For all $v\in V$, define $x_{u_i, v} = g_i(\alpha_{u_1}, \alpha_{u_2}, \ldots \alpha_{u_k})$, where $u_1, u_2, \ldots, u_k$ is the neighborhood of $v$. Then $x$ defines a fractional matching on $H$ with weight $n$.
\end{lemma}

Interested readers can find the proof in section \ref{sec:lemma-vee}.
Given this reconstruction technique, we can now describe the algorithm:
\begin{itemize}
\item In the preprocessing phase, find the $\alpha_u$ values for all $u\in U$ using Lemma \ref{lem:vee}.
\item In the online phase, for each online node $v$, compute $\tilde{x}_{u_i,v} = g_i(\alpha_{u_1}, \ldots, \alpha_{u_k})$, where $u_1, \ldots, u_k\in\Gamma_G(v)$, as described above.
Assign weight $\tilde{x}_{u_i,v}$ to the edge from $u_i$ to $v$; if that would cause node $u_i$ to have more than a total weight 1 assigned to it, just assign as much as possible. 
\end{itemize}
Note that we make the online computation based on the neighborhood in $G$, the realized graph, although the $\alpha_u$ values were computed based on $H$, the predicted graph.
We have the following.
\begin{theorem}\label{thm:fractional-matching}
In the $(d, \epsilon)$ semi-online matching problem in which the predicted graph has a perfect matching, there is a deterministic agnostic algorithm that gives a fractional matching of size $n(1-2\epsilon - \delta)$ in expectation, taken over the randomized realization of the graph. The algorithm does not know the value of $d$ or the value of $\epsilon$ in advance.
\end{theorem}
\begin{proof}
If the realized graph were exactly as predicted, we would give the fractional assignment $x$ guaranteed in Lemma~\ref{lem:vee}, which has weight $n$. However, the fractional matching that is actually realized is somewhat different.
For each online node that arrives, we treat it the same whether it is adversarial or not. But we have a few cases to consider for analysis:
\begin{itemize}
\item Case 1: The online node $v$ is adversarial. In this case, we forfeit the entire weight of 1 in the matching. We may assign some fractional matching to incident edges. However, we count this as `excess' and do not credit it towards our total. In this way, we lose at most $\delta n$ total weight.
\item Case 2: The online node $v$ is not adversarial, but it has extra edges added through a random process. There are at most $\epsilon n$ such nodes in expectation. In this case, we treat them the same as adversarial. We forfeit the entire weight of 1, and ignore the `excess' assignment. This loses at most $\epsilon n$ total weight in expectation.
\item Case 3: The online node $v$ is as exactly as predicted. In this case, we correctly calculate $x_{uv}$ for each $u\in\Gamma(v)$. Further, we assign $x_{uv}$ to each edge, unless there was already `excess' there. Since we never took credit for this excess, we will take $x_{uv}$ credit now. So we do not lose anything in this case.
\item Case 4: The online node $v$ is as predicted, except each edge is removed with probability $\epsilon$ (and no edges are added). In this case, when we solve for $z$, we find a value that is bounded above by the true $z$. The reason is that in the predicted graph, we solved $\sum_{u\in\Gamma(v)} \min(\max(0, \alpha_u - z), 1) = 1$ for $z$ when computing $g$.
In the realized graph, this same sum has had some of its summands removed, meaning the solution in $z$ is at most what it was before. So the value of $\tilde{x}_{uv}$ that we calculate is at least $x_{uv}$ for all $u$ in the realized neighborhood. We take a credit of $x_{uv}$ for each of these, leaving the rest as excess. Note that we have assigned $0$ to each edge that was in the predicted graph but missing in the realized graph. Since each edge goes missing with probability $\epsilon$, this is a total of at most $\epsilon n$ in expectation.
\end{itemize}
So, the total amount we lose in expectation is $2\epsilon n + \delta n$. Since the matching in the predicted graph has weight $n$, the claim follows.
\end{proof}

\subsection{Semi-Online Algorithms For Ski Rental}
\label{sec:ski-rental}

In this section, we consider the semi-online ski rental problem. In the classical ski rental problem, a skier needs to ski for an unknown number of days and on each day needs to decide whether to rent skis for the day at a cost of 1 unit, or whether to buy skis for a higher cost of $b$ units and ski for free thereafter. We consider a model where the skier has perfect predictions about whether or not she will ski on a given day for a few days in the time horizon. In addition, she may or may not ski on the other days. For instance, say the skier knows whether or not she's skiing for all weekends in the season, but is uncertain of the other days. The goal is to design an algorithm for buying skis so that the total cost of skiing is competitive with respect to the optimal solution for adversarial choices for all the days for which we have no predictions.

Let $x$ denote the number of days that the predictions guarantee the skier would ski. Further, it is more convenient to work with the fractional version of the problem so that it costs 1 unit to buy skis and renting for $z$ (fractional) days costs $z$ units. In this setting, we know in advance that the skier will ski for at least $x$ days. 
There is a randomized algorithm that guarantees a competitive ratio of $1/(1 - (1-x)e^{-(1-x)})$. Our analysis is a minor extension of an elegant result of~\cite{karlin2003stories}.

\begin{theorem}
\label{thm:ski-rental}
There is a $\dfrac{e}{e - (1-x)e^x}$ competitive randomized algorithm for the semi-online ski-rental problem where $x$ is a lower bound of the number of days the skier will ski.
\end{theorem}

\begin{proof}
Without loss of generality, we can assume that all the days with a prediction occur before any of the adversarial days arrive. Otherwise, the algorithm can always pretend as if the predictions have already occurred, since only the number of skiing days is important and not their order. Recall that $x$ denotes the number of days that the predictions guarantee the skier would ski. 
Let $u \geq x$ be the actual number of days (chosen by the adversary) that she will ski. Since buying skis costs $1$, the optimal solution has a cost of $\min(u, 1)$. Clearly, if $x \geq 1$, we must always buy the skis immediately and hence we assume that $0 \leq x < 1$ in the rest of the section. Further, even the optimal deterministic algorithm buys skis once $z=1$, so we may assume that $u\leq 1$.

\newcommand{\cost}{\ensuremath{\mbox{Cost}}}
Let $p_x(z)$ denote the probability that we buy the skis on day $z$, and let $q(x)$ denote the probability that we buy skis immediately. Recall that $p_x$ is implicitly a function of the prediction $x$. Given a fixed number of days to ski $u$, we can now compute the expected cost of the algorithm as
$$
\cost(x, u) = q(x) + \int_0^u (1+z) \cdot p_x(z) dz 
+ \int_u^1 u \cdot p_x(z) dz
$$
Our goal is to choose a probability distribution $p$ so as to minimize $\cost(x, u) / \min(u,1)$ while the adversary's goal is to choose $u$ to maximize the same quantity. We will choose $p_x$ and $q$ so that $\cost(x, u) / \min(u,1)$ is constant with respect to $u$. As we noted, $u\leq 1$, so $\min(u,1) = u$. Setting the $\cost(x,u) = c\cdot u$ for constant $c$ and taking the derivative with respect to $u$ twice gives us
$$
0 = \frac{\partial}{\partial u} p_x(u) - p_x(u)
$$
Of course, $p_x$ must also be a valid probability distribution.
Thus, we set $p_x(z) = (1-q(x))\cdot \dfrac{e^z}{e-e^x}$ for $z \geq x$. For $z < x$, we set $p_x(z) = 0$ since there is no reason to buy skis while $z<x$ if we did not already buy it immediately.

Recalling that we set $\cost(x,u) = c\cdot u$, we can substitute $p_x(z)$ and solve for $q(x)$, finding
\begin{align*}
q(x) &= \dfrac{x e^x}{e - (1-x)e^x}
\end{align*}
Hence, the competitive ratio is thus given by
$$
\frac{\cost(x,u)}{u} = 
	\frac{1}{u}\bigg(q(x) 
    + \frac{1-q(x)}{e-e^x}\int_x^u (1+z)e^z dz
    + \frac{1-q(x)}{e-e^x}\int_u^1 u\cdot e^z dz
    \bigg)
$$
Substitute $q(x)$, and after some manipulation, this becomes
$$
\frac{\cost(x,u)}{u} = 
	\frac{e}{e-(1-x)e^x}
$$
Note that when $x=0$, this becomes the classic ski rental problem, and the above bound is $e/(e-1)$, as expected.
\end{proof}

\appendix

\section{Proof of Lemma~\ref{lem:vee}}\label{sec:lemma-vee}
\begin{proof}


Recall that we are given a graph $H = (U, V, E)$, where $|U| = |V| = n$, which has a perfect matching. We may write the fractional matching problem as the following quadratic program.
\begin{align}
&\mbox{\ \ }\min \frac{1}{2}\sum_{(u,v)\in E} x_{uv}^2 \nonumber \\
\text{s.t. \ \ } 
	& \forall_u,\ \ \sum_{v\in\Gamma(u)} x_{uv} = 1 
    		\label{eqn:alpha}\\
	& \forall_v,\ \ \sum_{u\in\Gamma(v)} x_{uv} \leq 1
    		\label{eqn:beta}\\
    & \forall_{(u,v)\in E},\ \ \ 0 \leq x_{uv} \leq 1
    		\label{eqn:tau}
\end{align}
The objective here may be a little surprising. However, Constraint~\ref{eqn:alpha} guarantees that any feasible solution will have a matching of size $n$. The objective is chosen simply to make reconstructing the solution easy. We will see this in a moment.

Note that since there is a perfect matching, this optimization problem always has a feasible solution. Let's write out its KKT conditions. Let $\alpha_u$ be the dual for Constraint~\ref{eqn:alpha}, $\beta_v$ the dual for Constraint~\ref{eqn:beta}, and $\tau_{uv}^0, \tau_{uv}^1$ be the duals for Constraint~\ref{eqn:tau}. Letting $x^*$ denote the optimal solution, we have
\begin{align}
\forall_{(u,v)\in E},\ x^*_{uv}  - \alpha_u + \beta_v - \tau_{uv}^0 + \tau_{uv}^1 = 0 
	\label{eqn:primal}\\
\intertext{Complementary slackness:} 
\forall_u,\ 	\alpha_u \cdot \bigg(\sum_{v\in\Gamma(u)} x^*_{uv} - 1\bigg) = 0\\
\forall_v,\	\beta_v \cdot \bigg(\sum_{u\in\Gamma(v)} x^*_{uv} - 1\bigg) = 0
    			\label{eqn:comp_beta}\\
\forall_{(u,v)\in E},\	\tau_{uv}^0 \cdot x^*_{uv} = 0 = \tau_{uv}^1\cdot (x^*_{uv} - 1)
    			\label{eqn:comp_tau}\\
\text{with \ } \alpha_u, \beta_v, \tau_{uv}^0, \tau_{uv}^1 \ge 0\nonumber
\end{align}
Constraint~\ref{eqn:primal} can be rewritten as $x^*_{uv}  = \alpha_u - \beta_v + \tau_{uv}^0 - \tau_{uv}^1$. Inspecting this more carefully, note that if $0 < \alpha_u - \beta_v < 1$, then by complementary slackness (Constraint~\ref{eqn:comp_tau}), both $\tau_{uv}^0 = \tau_{uv}^1 = 0$. That is, $x^*_{uv} = \alpha_u - \beta_v$. On the other hand, if $\alpha_u - \beta_v < 0$, then since $x^*_{uv} \geq 0$, we must have $\tau_{uv}^0 = \beta_v - \alpha_u$. Hence, $x^*_{uv} = 0$. An analogous argument shows that if $\alpha_u - \beta_v > 1$, then $x^*_{uv} = 1$. Putting this together, we see
$$
x^*_{uv} = \min(\max(0, \alpha_u - \beta_v), 1)
$$
In other words, we can reconstruct $x^*_{uv}$ using only the $n$ dual $\alpha$ values and $n$ dual $\beta$ values. Wonderfully, we can do even better. Suppose we have only the $\alpha$ values. Given the neighborhood of any online node $v$, we can reconstruct the value of $x^*_{uv}$ for all $u$ in the neighborhood. By Constraint~\ref{eqn:comp_beta}, we have
\begin{align*}
\beta_v \cdot \bigg(\sum_{u\in\Gamma(v)} x^*_{uv} - 1\bigg) = 0\\
\beta_v \cdot \bigg(\sum_{u\in\Gamma(v)} \min(\max(0, \alpha_u - \beta_v), 1) - 1\bigg) = 0 
\end{align*}
Since we know all the $\alpha_u$, and $x_{uv} \leq 1$, we can solve this for $\beta_v$ (in linear time). Just solve the equation $\sum_{j\in\Gamma(i)} \min(\max(0, \alpha_u - z), 1) = 1$ for $z$ and set $\beta_v = \max(0, z)$.
In other words, given only the $\alpha$ dual values, we can reconstruct all the $x^*_{uv}$ using the structure of the graph.

Noting that we can solve this quadratic program in polynomial time (using the Ellipsoid Algorithm), we can find the dual values $\alpha_u$ in polynomial time, and Lemma~\ref{lem:vee} follows.
\end{proof}

\bibliography{bibfile}

\end{document}